\documentclass[11pt,reqno]{amsart}
\usepackage{amsthm}
\usepackage{enumerate}
\usepackage{amssymb}
\usepackage{amsthm}
\usepackage[T1]{fontenc}
\usepackage{ae}
\usepackage{amsfonts}
\usepackage{hyperref}
\usepackage{graphicx}
\usepackage[table]{xcolor}
\usepackage{booktabs,caption,float,multirow,siunitx,tabularx}

\theoremstyle{plain}
\newtheorem{Thm}{Theorem}[section]

\newtheorem{Def}[Thm]{Definition}
\newtheorem{Exm}[Thm]{Example}

\usepackage[foot]{amsaddr}

\renewcommand{\baselinestretch}{1.1}\small\normalsize
\numberwithin{equation}{section}

\usepackage[latin1]{inputenc}
\usepackage{times}
\usepackage{tikz}
\usepackage{amsmath}
\usepackage{verbatim}
\usetikzlibrary{arrows,shapes}

\usepackage{graphics}                 % Packages to allow inclusion of graphics
\usepackage{color}                    % For creating coloured text and background
\usepackage{hyperref}                 % For creating hyperlinks in cross references
\usepackage{graphicx}

\usepackage{pgf}
\usepackage{tikz}
\usetikzlibrary{arrows,automata}
\usepackage[latin1]{inputenc}
\usepackage{verbatim}

\begin{document}

\title[Continuous and discrete dynamics of HIV infection]
{Continuous and discrete dynamics of  a deterministic model of HIV infection}

\author{}
\address{$^*$Majid Jaberi Douraki
% D\'epartement de math\'ematiques et de
%statistique,
%         Universit\'e Laval,
%         Qu\'ebec (QC),
%         Canada G1K 7P4.\hfill\break
         }
\email{Jaberi@K-State.edu}

%\author{}
%\address{Javad Mashreghi\hfill\break
%D\'epartement de math\'ematiques et de statistique,
%         Universit\'e Laval,
%         Qu\'ebec(QC),
%         Canada G1K 7P4.}
%\email{Javad.Mashreghi@mat.ulaval.ca}
%
%\author{}
%\address{Mostafa Nasri\hfill\break
%Instituto de Matematica Pura e Aplicada,
%Estrada Dona Castorina 110,
%Jardim Botanico  22460-320,
%Rio de Janeiro  RJ  BRAZIL }
%\email{mostafa@impa.br}

%\author[M. Jaberi Douraki, J. Mashreghi, and M. Nasri]
%{M. Jaberi Douraki, J. Mashreghi, and M. Nasri}

\author[M. Jaberi Douraki]
{{\bf Majid Jaberi Douraki$^*$}\\
Institute of Computational Comparative Medicine,\\
Department of Mathematics,\\
Kansas State University,
Manhattan, KS 66506-5802}

%\thanks{This work was supported by NSERC (Canada) and FQRNT (Qu\'ebec)}


\begin{abstract}
%A major pharmaceutical intervention for management of HIV (human immunodeficiency virus) infection is the use of anti-HIV preventive vaccines
%(prophylactic vaccines). Unlike the rise of several treatments, the HIV infection still poses significant threats to the human beings. This research proposes an
%optimal treatment therapy which can minimize the viral load through an effective combination of reverse transcriptase and protease inhibitors in vivo.
We will study a mathematical model of the human immunodeficiency virus (HIV) infection in the presence of combination therapy that includes within-host
infectious dynamics.
%in the context of disease spread in vivo.
%
%Determining optimal treatment strategies is generally a challenging problem and requires the application of more sophisticated mathematical techniques, such as the use of discrete dynamical system.
The deterministic model requires us to analyze asymptotic stability of two distinct steady states, disease-free and endemic equilibria. Previous results have focused on investigating the global asymptotic stability of the trivial steady state using an implicit finite-difference method which generates a system of difference equations. We, instead, provide analytic solutions and long term attractive behavior for the endemic steady state using the theory of difference equations. The dynamics of estimated model is appropriately determined by a certain quantity threshold maintaining the immune response to a sufficient level. The result also indicates that a forward bifurcation in the model happens when the disease-free equilibrium loses its stability and a stable endemic equilibrium appears as the basic reproduction number exceeds unity. In this scenario, the classical requirement of the reproduction number being less than unity becomes a necessary and sufficient condition for disease mitigation. When the associated reproduction number is in excess of unity, a stable endemic equilibrium emerges with an unstable disease-free equilibrium (leading to the persistence and existence of HIV within the infected individuals). The attractivity of the model reveals that the disease-free equilibrium is globally asymptotically stable under certain assumptions.
A comparison between the continuous and estimated discrete models is also provided to have a clear perception in understanding the behavioral dynamics of disease modelling. Finally, we show that the associated estimation method is very robust in the sense of numerical stability since the equilibria and the stability conditions are  independent of the time step.

%Nonetheless, in the presence of $60\%$ effective combination therapy, the unique endemic equilibrium is globally asymptotically stable.
\end{abstract}

\maketitle

\section{Introduction}
Although rising number of anti-HIV (human immunodeficiency virus) preventive vaccines (also prophylactic vaccines) currently undergo clinical trials (see \cite{chap7_Blower} and the references therein), there has been a surge of interest during the past decades with regards to mathematical model for HIV infection in order to determine the side effects and long term impacts of such vaccines to the community. Identifying an optimal treatment therapy to minimize the viral load still poses significant threats to the human beings and a dilemma for public health policy maker \cite{JaberiJBD,JaberiMBE,JHWM,Lin}. In particular, since none of the current anti-HIV vaccines attains high enough efficacy to prevent HIV infection, such models may give insight into the optimal efficacy and treatment coverage levels needed to mitigate HIV infection in the community \cite{chap7_Blower}.

%\textbf{Targeting HIV replication}

%the development of the virion such as reverse transcriptase, protease, ribonuclease and integrase

Although some  published papers concerning the HIV infection models provides mostly slight details on
the system of difference equations extracted from associated systems of differential equations,  authors have thoroughly studied this subject
\cite{chap7_D:Clark, chap7_M. Dehghan 2004-1, chap7_M. Fan, chap7_M. Jaberi, chap7_Kouichi, chap7_R.E. Mickens,
chap7_M. Nasri}. Besides, there are outstanding published papers in the mathematical literature in the field of difference equations while a few of
them discussed systems of difference equations with two equations or more than one equilibrium point \cite{chap7_F. Brauer, chap7_V:Kocic, chap7_G:Papaschinopoulos, chap7_R.B. Potts, chap7_H. Sedaghat 2003}. Implementation of control theory to infectious diseases is another challenge to alleviate infection widespread. For instance, Jaberi et al. \cite{JHWM} developed an optimal control theory to a population model for influenza infection with sensitive and resistant strains and solved the adjoint control components to find the optimal treatment profile that reduces the epidemic final size. In the current research work, however, we attempt to analyze and mitigate a HIV infection presented by a model of difference equations with three equations associated with two distinct equilibrium points.

Apart from impede and enhance performance in mathematical analysis, several challenges related to the characteristics of the pathogen, e.g.  replication and evolution, have been encountered the nature of host immune responses at the individual level \cite{chap7_Blower,chap7_F. Brauer,chap7_P. Essunger, chap7_M. Farkas, chap7_A.B. Gumel,Mazloom}. The viral replication is essential for progression of HIV to AIDS and, as a result, requires better administration and intervention of antiretroviral drugs \cite{Kirschner1,Kirschner2,Kirschner3,McLean}. The replication process of the HIV pathogen undergo a cycle of development within a human host. Several essential phases (enzymes) such as reverse transcriptase, protease, ribonuclease and integrase are required for successful development of infectious virus particle \cite{chap7_Blower,chap7_F. Brauer,chap7_P. Essunger, chap7_M. Farkas, chap7_A.B. Gumel,Kirschner3,McLean}. In most cases these replication process are taken into consideration and mathematically formulated by different types of population models using differential equations.
Generally speaking, a handful of these models underwent analytic approaches to investigate the sensitivity and stability analysis. Other approaches, mostly computational, have been taken into account to determine the most efficient way of drug treatment at both the individual and population levels \cite{chap7_A.B. Gumel}. Not surprisingly, some mathematical models dealing with a nonlinear IVP (initial value problem) system do not have a closed form solution, nor can one, even locally, predict the long-term behavior of infection. As a result, a numerical method (e.g. see section \ref{sec4}) is widely applied to observe the characteristics of the system with various effects of time steps, especially for artificially large time steps. However, one may argue that if a numerical method is applied to a system of differential equations, the solution obtained through the numerical scheme may be inconsistent with the original system. An appropriate explanation to this argument will be provided in section \ref{subsec1}.

The outline of current paper is organized as follows. Our main purpose in this work is to study the stability and bifurcation of the existing
mathematical model for HIV infection
\begin{align}
T_{4}'(t)&= (s+r\,T_{4}(t)\,V_{I}(t))-k_{v}\,(1-E_{RT})\,(r+\alpha)\,T_{4}(t)\,V_{I}(t)-
\gamma_{1}\,T_{4}(t), & T_{4}(t_{0})=T_{4}^{0},\label{chap7_e2.1.n1}\\
T_{I}'(t)&= (k_{v}\,(1-E_{RT})\,(r+\alpha))\,T_{4}(t)\,V_{I}(t)- (\gamma_{2}+ k_{c})\,T_{I}(t),
  & T_{I}(t_{0})=T_{I}^{0},\label{chap7_e2.1.n2}\\
V_{I}'(t)&= \gamma_{2}\,(1-E_{PI})\,N\,(1-L)\,T_{I}(t)-\sigma\,
V_{I}(t),& V_{I}(t_{0})=V_{I}^{0},\label{chap7_e2.1.n3}
\end{align}
where the parameters are given in Table \ref{chap7_table0}. We will investigate the long term behavior of solutions in two different scenarios, continuous and discrete forms. A system of difference equations will be obtained from discretization of the mathematical model. It will be shown that the estimation method used in \cite{chap7_A.B. Gumel} to discretize system (\ref{chap7_e2.1.n1})--(\ref{chap7_e2.1.n3}) is significantly compatible with the continuous system and is a robust numerical method to investigate the nature of the IVP system (for instance, see \cite{chap7_R.E. Mickens}), that is, we can obtain similar properties (such as boundedness, asymptotic stability, oscillatory behavior) of the IVP system from the associated numerical scheme. It is worthwhile to point out that we demonstrate  the condition for the endemic equilibrium point of the discrete model to be stable is given by $e>(b-1)(1-d)(1-f)$ which was conjectured in \cite{chap7_DNJ} as an Open Problem, i.e. the associated reproduction number is greater than unity for stable endemic state. Similarities and differences between two models are reported in the beginning of Section \ref{sec5}. This research ends with extensive simulations and a brief conclusion.

\section{A survey on mathematical model and dynamics of continuous system}\label{sec2}
The viral replication is essential for progression of HIV to AIDS and, as a result, requires better administration and intervention of antiretroviral drugs \cite{Kirschner1,Kirschner2,Kirschner3,McLean}. The replication process of the HIV pathogen undergo a cycle of development within a human host. Several essential phases (enzymes) such as reverse transcriptase, protease, ribonuclease and integrase are required for successful development of infectious virus particle \cite{chap7_Blower,chap7_F. Brauer,chap7_P. Essunger, chap7_M. Farkas, chap7_A.B. Gumel,Kirschner3,McLean}. Each step is therefore a potential target for identification of an optimal treatment therapy. In the first step, infectious virus particles  locate appropriate host cells such as a CD4$^{+}$ T cells to infect. The process of infection begins by entering HIV into the target cell by fusion after binding to the CD4 glycoprotein and then releasing three crucial replication enzymes: Reverse Transcriptase, Integrase and Protease. Prophylactic medications which interfere with the principal stages of viral replication can inhibit this mortal contamination. For instance, introduction of HIV into the recipient cell which involves in initial infection can be possibly impeded by implementation of fusion inhibitors. Mitigation of reverse transcriptase by application of nucleoside inhibitors or by non-nucleoside Reverse Transcriptase  inhibitors is also part of standard therapeutic course of medical treatments. Besides, the fatal process of Retroviral Integrase which enables the genetic material of retrovirus enzyme to be integrated into the DNA of the infected cell can be blocked by HIV integrase inhibitors such as Raltegravir \cite{Steigbigel}. Perhaps the last key step in a standard antiretroviral therapy is to inhibit the retroviral aspartyl protease causing disease progression by intervention of protease  inhibitors including Saquinavir, Ritonavir, Indinavir, Nelfinavir, Amprenavir \cite{Rang}.

The model simulates the interaction between CD4$^{+}$ T-lymphocyte
(these cells perform essential helper functions in the
development of cellular and humoral immunity against pathogens
including HIV, see \cite{chap7_P. Essunger, chap7_M. Farkas, chap7_A.B. Gumel}) and
HIV in vivo when combination antiretroviral therapy is used.
Namely, reverse transcriptase (RT) inhibitors and protease
inhibitors (PIs)  as antiviral drug intervention are exerted for the perturbation of HIV infection.

\begin{table}[hbt]\label{chap7_table_1}
\tikzstyle{int}=[rectangle, draw, fill=blue!10,
    text width=22em, text centered, rounded corners, minimum height=10.5em]
\tikzstyle{init} = [pin edge={to-,thick,black}]
{
\tikzstyle{cloud} = [draw=none, rectangle,text centered,rounded corners,text width=\textwidth,fill=gray!20, node distance=3.8cm,
    minimum height=5em]%;
\begin{tikzpicture}[node distance=2.cm,auto,>=latex']
    \node [cloud] (dd)
    {\begin{center}\footnotesize
    \begin{tabularx}{\textwidth}{llXl}
            Parameter     && Interpretation     & Value\\[1ex]\hline
   $s$           && Source term for uninfected CD4$^{+}$ T cells & $8.076 \,\,d^{-1} mm^{-3}$\\
   $r$           && Rate of proliferation of CD4$^{+}$ T cells   &   $0.03 \,\,d^{-1}$\\
   $k_{v}$       && Probability of infection of activated CD4$^{+}$ T cells  & 1\\
   $\gamma_{1}$  && Natural death rate of uninfected CD4$^{+}$ T cells   &   $0.5 \,\,d^{-1}$\\
   $\gamma_{2}$  && Natural death rate of infected CD4$^{+}$ T cells   &   $0.5 \,\,d^{-1}$\\
   $k_{c}$       && Anti-HIV immune response             &       0.5\\
   $N$           && Number of free viruses produced per infected cell  &   1000\\
   $L$           && Proportion of latently infected T cells      &   0.2\\
   $E_{RT}$      && Effectiveness of RT inhibitors           &   $0<E_{RT}<1$\\
   $E_{PI}$      && Effectiveness of protease inhibitors       &   $0<E_{PI}<1$\\
   $\sigma$      && Rate of clearance of infectious virions    & 10 $d^{-1}$\\
   $\alpha$      && Proportion of pre-existing activated CD4$^{+}$ T cells &   $0<\alpha<1$\\[1ex]
   \hline
   \end{tabularx}
\caption{\footnotesize Description of parameters in HIV model and their estimated values \cite{chap7_M. Dehghan 2007,chap7_A.B. Gumel,chap7_DNJ}}\label{chap7_table0}
\end{center}};
\end{tikzpicture}
}
\end{table}
\begin{figure}[hbtp]
%\begin{center}
%  \begin{tabular}{cc}
%     {\fbox  {\includegraphics*[height=3.in, width=6in]{Compartmental_model.eps}}}
%  \end{tabular}
\tikzstyle{int}=[rectangle, draw=none, fill=blue!10,
    text width=12em, text centered, rounded corners, minimum height=3.5em]
\tikzstyle{init} = [pin edge={to-,thick,black}]
{\footnotesize
\tikzstyle{cloud} = [draw=none, rectangle,text centered,rounded corners,text width=20em,fill=blue!10, node distance=3.8cm,
    minimum height=4em]%;
\begin{tikzpicture}[node distance=2.cm,auto,>=latex']
    \node [int] (a) {CD4+T Cells (Susceptible)};
    \node (b) [above of=a,node distance=1.2cm, coordinate] {a};
    \node [int] (c) [below of=a] {Infected cells carrying integrated HIV (Infected)};
    \node [int] (d) [below of=c] {Infectious virus particles  (Infectious)};
    \node [coordinate] (end1) [right of=a, node distance=2.75cm]{};
    \node [coordinate] (end2) [right of=c, node distance=2.75cm]{};
    \node [coordinate] (end3) [right of=d, node distance=2.75cm]{};
    \node [cloud] (dd) [right of=end2,node distance=4.5cm] {
    \begin{eqnarray*}
          \frac{dT_{4}}{dt}&=& (s+r\,T_{4}\,V_{I})-\lambda_{1}\,T_{4}\,V_{I}- \gamma_{1}\,T_{4},\label{chap7_e2.1.1}\\\\
          \frac{dT_{I}}{dt}&=& \lambda_{1}\,T_{4}\,V_{I}- (\gamma_{2}+ k_{c})\,T_{I},\label{chap7_e2.2.1}\\\\
          \frac{dV_{I}}{dt}&=& \lambda_{2}\, T_{I}-\sigma\,V_{I},\label{chap7_e2.3.1}
    \end{eqnarray*}};
    {\path[->] (b) edge node {Source} (a);}
    \hspace{2cm}{\path[->] (a) edge node[left] {$\lambda_{1}=(r+\alpha)(1-E_{RT})k_v$} (c);}
    \hspace{-2cm}\draw[->] (a) edge node {$\gamma_1$} (end1) ;
    \draw[->] (c) edge node {$k_c$} (end2) ;
    \draw[->] (d) edge node {$\sigma$} (end3) ;
    \hspace{2.2cm}\draw[->] (c) edge node[left] {$\lambda_2=\gamma_2(1-E_{PI})N(1-L)$} (d) ;
%    \node[cloud]    (Ee) [right  of=a, node distance=3.2cm] {asdfds}
\end{tikzpicture}
}
\caption{\footnotesize Compartmental model of HIV infection for three different cell types, i.e. susceptible, infected, and infectious. In this model, $\lambda_1$ and $\lambda_2$ are the force of infection from the  susceptible compartment to the infected one and  the force of infection from the infected compartment to the infectious one, respectively.}
  \label{chap7_Three Cases}
%\end{center}
\end{figure}
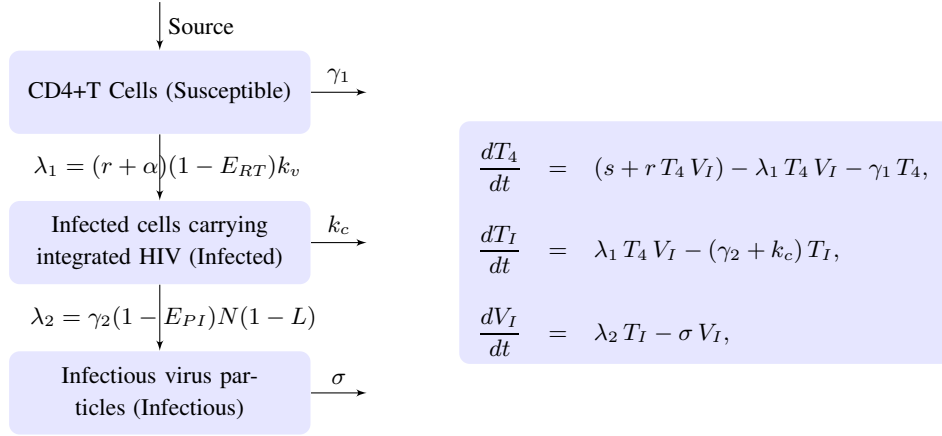

The existing mathematical model will be  briefly described in the sequel, for further information and the terminology used here, we refer readers to \cite{chap7_A.B. Gumel, chap7_M. Nasri, chap7_DNJ, chap7_A.S. Perelson}. The experiment shows that the reciprocal interplay between infectious viral particles and
CD4$^{+}$ T cells produces three sizable quantities in vivo which are characterized by the mathematical notations: $T_{4}$,
$T_{I}$, and $V_{I}$, where the variable $T_{4}$ is used to illustrate the amount of inactivated CD4$^{+}$ T cells at each unit of time which is
not accompanied by  integrated HIV (or susceptible cells in patient), the second quantity $T_{I}$ is taken into consideration to represent the amount of CD4$^{+}$ T cells carrying integrated HIV (infected cells), and the third class $V_{I}$ indicates infectious HIV
load (infectious particles) released from infected cells. As depicted in Figure \ref{chap7_Three Cases}, a compartmental model is used to facilitate the flow in each category and the interaction between CD4$^{+}$ T cells and viral particles with the intervention of combination antiretroviral therapy. Using the three compartments illustrated for the sizable amount of susceptibles, infected cells and free virions, we obtain the following associated  IVP system  at any time $t$, $t>t_{0}$,
\begin{align}
 \frac{dT_{4}}{dt}&= (s+r\,T_{4}(t)\,V_{I}(t))-\lambda_{1}\,T_{4}(t)\,V_{I}(t)- \gamma_{1}\,T_{4}(t),\nonumber\\
                  &= s+e_{1}\,T_{4}(t)\,V_{I}(t)-\gamma_{1}\,T_{4}(t), & T_{4}(t_{0})=T_{4}^{0},\label{chap7_e2.1}\\
\frac{dT_{I}}{dt}&= e_{2}\,T_{4}(t)\,V_{I}(t)- e_{3}\,T_{I}(t),& T_{I}(t_{0})=T_{I}^{0},\label{chap7_e2.2}\\
\frac{dV_{I}}{dt}&= e_{4}\,T_{I}(t)-\sigma \,V_{I}(t),& V_{I}(t_{0})=V_{I}^{0},\label{chap7_e2.3}
\end{align}
where the recruitment term $(s+r\,T_{4}(t)\,V_{I}(t))$ represents the autocatalysis (autocatalytic reaction in a chemical reaction was also examined by Petrov et al. \cite{Petrov}, see also \cite{Murray}, with cubic and quadratic autocatalysis term) and the parameters are given in Table \ref{chap7_table0} and defined by
\begin{equation}\label{chap7_e2.4}
  \begin{array}{lrl}
  e_{1}=r-k_{v}\,(1-E_{RT})\,(r+\alpha) &\mbox{and}& e_{2}=k_{v}(1-E_{RT})\,(r+\alpha);\\
  e_{3}=\gamma_{2}+ k_{c} &\mbox{and}& e_{4}=\gamma_{2}\,(1-E_{PI})\,N\,(1-L).\\
  \end{array}
\end{equation}
Note that the infection term is given by $\lambda_1$ in Figure \ref{chap7_Three Cases}, not by $e_1$. As a result, the law of conservation of mass will remain constant over time,  that is, the population mass leaving the susceptible (healthy) compartment is taken into account in the succeeding compartment regarded as infected cells carrying integrated virus particles. For the simplicity in analysis of system (\ref{chap7_e2.1})$-$(\ref{chap7_e2.3}) and comparative presentation, we add together $r$ and $k_{v}\,(1-E_{RT})\,(r+\alpha)$ and use the notation $e_{1}$.

The following change of variables simplifies the model (\ref{chap7_e2.1})$-$(\ref{chap7_e2.3}),
where $e_1\neq0$:
\begin{equation}\label{chap7_eq1}
T_{4}=\frac{s}{\gamma_1}\,x,~~~\mbox{}~~~
T_I=\dfrac{e_2\,s}{|e_1|\,\gamma_1}\,y,~~~\mbox{and}~~~V_{I}=\dfrac{e_2\,e_4\,s}{|e_1|\,e_3\,\gamma_1}\,z.
\end{equation}
In fact, the system of differential equations (\ref{chap7_e2.1})$-$(\ref{chap7_e2.3})
takes the following form:
\begin{align}
x'&= F(x,z)= \text{sign}(e_1) \,\rho \,x\,z+\tau-\tau\, x, &x(0)=x_0,\nonumber\\
y'&= G(x,y,z)= \rho \,x\,z- \zeta\, y, &y(0)=y_0,\label{chap7_eq2}\\
z'&= H(y,z)= \zeta\, y-\eta\, z, &z(0)=z_0,\nonumber
\end{align}
in which $\tau=\gamma_1$, $\rho=\dfrac{e_2\,e_4\,s}{e_3\,\gamma_1}$,
$\zeta=e_3$,  $\eta=\sigma$, and sign($\cdot$) is the  signum function of real numbers.  It is worthwhile to mention that solutions of
system (\ref{chap7_eq2})
is positively invariant with respect to the first octant, i.e. $\mathbb{R}^3_{+} = \{(x,y,z)| x\geq 0, y\geq 0, z\geq 0\}$, \cite{chap7_M. Dehghan 2007}. The equilibria of  system (\ref{chap7_eq2}),
\begin{equation}\label{chap7_eq3}
    E_1=(1,0,0) \quad \mbox{and} \quad
E_2=\left(\dfrac{\eta}{\rho},\dfrac{\tau\,(\eta-\rho)}{\text{sign}(e_1)\,\rho\,\zeta},\dfrac{\tau\,(\eta-\rho)}{\text{sign}(e_1)\,\rho\,\eta}\right).
\end{equation}
can easily be evaluated by some algebraic calculation. Now set $\mathcal{R}_c={\rho}/{\eta}$. The quantity $\mathcal{R}_c$ is the basic reproduction number for the continuous model (\ref{chap7_eq2}).
Note that if $\mathcal{R}_c=1$, then the system (\ref{chap7_eq2}) has no
equilibria other than $E_1$.
To investigate the local asymptotic stability, we evaluate the
characteristic polynomial by first differentiating the right-hand side of (\ref{chap7_eq2}) and then calculating the Jacobian matrix at the equilibria
\cite{chap7_M.W. Hirsch 1974,chap7_R.K. Miller 1982,chap7_L. Perko 1991}. The
characteristic polynomial at the disease-free equilibrium, $E_1$, is given by $-(\lambda+\tau)\left[\lambda^{2}+(\zeta+\eta)\lambda+\zeta(\eta-\rho)\right]$, and the characteristic polynomial at the endemic state,  $E_2$, is presented by
\begin{equation}\label{chap7_eq6}
-\lambda^{3}-\left(\zeta+\eta+\dfrac{\tau\,\rho}{\eta}\right)\,\lambda^{2}-\left(\tau\,\rho+\dfrac
{\tau\,\rho\,\zeta}{\eta}\right)\,\lambda-\tau\,\zeta\,(\rho-\eta).
\end{equation}
Using the Routh-Hurwitz criterion \cite{chap7_R.K. Miller 1982} for the above characteristic polynomials, one can show \cite{chap7_M. Dehghan 2007,chap7_A.B. Gumel} that the disease-free state $E_1$ is locally asymptotically stable if $\mathcal{R}_c<1$ and is unstable if $\mathcal{R}_c>1$. Also for the latter case, the endemic state $E_2$ is locally asymptotically stable and is unstable if $\mathcal{R}_c<1$. In addition, the system has no equilibria other than $E_1$ when $\mathcal{R}_c=1$. The eigenvalues of the Jacobian matrix at the disease-free state are $\{0, -\tau,-\zeta-\eta\}$.
It was shown (see Theorem 3.4 in \cite{chap7_M. Dehghan 2007}) that if
$e_1>0$ and $\mathcal{R}_c>1$, it follows that $(x,y,z)\rightarrow(\infty,\infty,\infty)$ for $(x_0,y_0,z_0) \in \Delta=\{(x,y,z)~|~y\neq0 \mbox{ or } z\neq0 \} \cap
\mathbb{R}^3_{+}$.
%with
%\begin{equation}\label{chap7_eq8}
%\Delta=\{(x,y,z)~|~y\neq0 \mbox{ or } z\neq0 \} \cap
%\mathbb{R}^3_{+}.
%\end{equation}
%In view of (\ref{chap7_eq8}), we mean that there exists at least a small quantity (unit) of infected or infectious cell in vivo.
It was also illustrated by some experiments in \cite{chap7_M. Dehghan 2007} that
system (\ref{chap7_eq2}) has bounded and unbounded solutions depending
on the domain of initial values when $e_1>0$ and $\mathcal{R}_c\ge1$, while the following result shows the stability of steady states when $e_1< 0$ (see Theorems
4.1, 4.11 and Remark 4.2 in \cite{chap7_M. Dehghan 2007}).

\begin{Thm}\label{chap7_th2}
Consider system (\ref{chap7_eq2}) with $e_1<0$. Then the following statements are true.
\begin{itemize}
\item[(i)]
If the reproduction number $\mathcal{R}_c$ is less than or equal to unity, then the disease-free state is globally
asymptotically stable in $\mathbb R_+^{3}$. In this case, HIV infection is eradicated from the infected host.
\item[(ii)]
If  the reproduction number $\mathcal{R}_c$ is greater than unity, then the endemic state is globally asymptotically
stable in $\Delta$.
%where $\Delta$ is given by (\ref{chap7_eq8}).
This case leads to the persistence and existence
of HIV infection within the individual.
\end{itemize}
\end{Thm}

Considering system (\ref{chap7_e2.1})$-$(\ref{chap7_e2.3}) with $e_1=0$, it follows that the equilibrium point $(s/\gamma_{1},0,0)$ is unique when $w=C\,e_3\,\gamma_1-s\,e_2\,e_4\neq0$. When $w=0$, the model (\ref{chap7_e2.1})$-$(\ref{chap7_e2.3}) also possesses the unique positive equilibrium $(s/\gamma_{1},\,e_{4}\,\theta/C,\theta)$, in which $\theta$ depends critically on the coefficients of
the system and the initial value $(x_{0},y_{0},z_{0})$. The following theorem  shows the qualitative behavior of
 system (\ref{chap7_e2.1})$-$(\ref{chap7_e2.3}) when $e_1=0$.% and $w\geq 0$.

\begin{Thm}[See \cite{chap7_M. Dehghan 2007}]\label{chap7_th3}
Consider  system (\ref{chap7_e2.1})$-$(\ref{chap7_e2.3}) with $e_1=0$. Then the
following statements are true.
\begin{itemize}
\item[(i)]
If $w>0$, then the disease-free state is globally
asymptotically stable in $\mathbb R_+^{3}$.
\item[(ii)]
If $w=0$, then positive solutions converges eventually to  $(s/\gamma_{1},\,e_{4}\,\theta/C,\theta)$ when $(T_4^0,T_I^0,V_I^0)\in\mathbb R_+^{3}$.
\end{itemize}
\end{Thm}

The model (\ref{chap7_e2.1})$-$(\ref{chap7_e2.3}) has unbounded solutions whenever $e_1=0$
and $w<0$ (see Remark 5.4 and examples illustrated in \cite{chap7_M. Dehghan 2007}).

%------------------------------------------------------------------------------
%------------------------------------------------------------------------------

\section{Analysis of discrete model}

\subsection{Preliminaries to difference equations}\label{sec3}
In this part, we begin by presenting some basic definitions to study the qualitative  behavior of a system of difference equations.
%One helpful supportive lemma which plays key role in the next section is  also stated
%here.
Let $I_{1}$, $I_{2}$, and $I_{3}$ be some intervals of real
numbers and let
\begin{equation*}
F_{i}\,:\, I_{1}\times I_{3}\times I_{3}\longrightarrow I_{i},\qquad \mbox{for $i\in\{1,2,3\}$}
\end{equation*}
be three continuously differentiable functions. Then for every {\em initial condition} $(x_{0},y_{0},z_{0})
\in I_{1}\times I_{2}\times I_{3}$, the system of difference
equations
\begin{eqnarray}
x_{n+1}=F_{1}(x_{n},z_{n}), \qquad y_{n+1}=F_{2}(x_{n+1},y_{n},z_{n}),\qquad z_{n+1}=F_{3}(y_{n+1},z_{n}) \label{chap7_e3.1}
\end{eqnarray}
has a {\em unique solution}
$\{(x_{n},y_{n},z_{n})\}_{n=0}^{\infty}$. A point $(\bar{x},\bar{y},\bar{z}) \in I_{1}\times I_{2}\times
I_{3}$ is called an {\em equilibrium point} of system
          (\ref{chap7_e3.1}) if
          \begin{eqnarray}
          \bar{x}=F_{1}(\bar{x},\bar{z}),\qquad\bar{y}=F_{2}(\bar{x},\bar{y},\bar{z}),\qquad\bar{z}=F_{3}(\bar{y},\bar{z}).\label{chap7_e3.4}
          \end{eqnarray}
\begin{Def}
The equilibrium point $(\bar{x},\bar{y},\bar{z})$ of system
(\ref{chap7_e3.1}) is called
\begin{enumerate}
  \item{\em stable} (or {\em
locally stable}) if for every $\epsilon>0,$ there exists $\delta>0$
such that for all $(x_{0},y_{0},z_{0}) \in I_{1}\times I_{2}\times
I_{3}$ with $\|(x_{0},y_{0},z_{0})-(\bar{x},\bar{y},\bar{z})\|<\delta$
implies $\|(x_{n},y_{n},z_{n})-(\bar{x},\bar{y},\bar{z})\|<\epsilon$ for all $n\geq0.$
Otherwise the equilibrium point is called unstable.

          \item {\em asymptotically stable} (or {\em locally asymptotically stable}) if it is stable and there exists $\gamma>0$ such that for all $(x_{0},y_{0},z_{0})$ $\in I_{1}\times I_{2}\times I_{3}$ with $\|(x_{0},y_{0},z_{0})-(\bar{x},\bar{y},\bar{z})\|<\gamma$ implies $\lim_{n\rightarrow\infty} \|(x_{n},y_{n},z_{n})-(\bar{x},\bar{y},\bar{z})\|=0.$

\item {\em globally asymptotically stable} if it is locally stable and for every $(x_{0},y_{0},z_{0})$ $\in
I_{1}\times I_{2}\times I_{3}$, we have $\|(x_{n},y_{n},z_{n}) - (\bar{x},\bar{y},\bar{z})\|\to0$, as $n\rightarrow\infty. $
\item a {\em repeller}, if there
exists $r>0$ such that for all $(x_{0},y_{0},z_{0})$ $\in
I_{1}\times I_{2}\times I_{3}$ with $0<\|(x_{0},y_{0},z_{0})-(\bar{x},\bar{y},\bar{z})\|<r,$ there exists $N\geq1$
such that $\|(x_{N},y_{N},z_{N})-(\bar{x},\bar{y},\bar{z})\|\geq r.$
\end{enumerate}
\end{Def}

\subsection{Discretizing the model}\label{sec4}
Here, we investigate the boundedness, the stability
analysis, and the bifurcation phenomenon of a discrete system in general cases. The development of the numerical method will be based on the first-order
approximation and using a numerical method of the form \cite{chap7_A.B. Gumel,chap7_R.E. Mickens,Faqir,Q_Zhang,P_Zhang1,P_Zhang2,JaberiJTB1,JaberiJTB2,JaberiPediatric,JaberiPlosOne}
\begin{equation}\label{e1.1}
W^{n+1}=W^{n}+l\,f(W^{n},W^{n+1}),\qquad n\ge0
\end{equation}
where $W^{n}=[T_{4}^{n}, T_{I}^{n},V_{I}^{n}]^{\top}$, ${\top}$ denoting transpose. The robust numerical method can be obtained by approximating the time derivative with its first-order forward-difference approximation for each equation in the left side of system
(\ref{chap7_e2.1})$-$(\ref{chap7_e2.3}) and making appropriate estimations for the right side, as follows with the fixed time step size $\ell>0$.
\begin{eqnarray}
({T_{4}^{n+1}-T_{4}^{n}})/\ell&=&s+e_{1}\,T_{4}^{n+1}\,V^{n}_{I}-\gamma_{1}\,T_{4}^{n+1},\label{chap7_e2.5}\\
({T_{I}^{n+1}-T_{I}^{n}})/{\ell}&=&e_{2}\,T_{4}^{n+1}\,V^{n}_{I}-e_{3}\,T_{I}^{n+1}, \label{chap7_e2.6}\\
({V_{I}^{n+1}-V_{I}^{n}})/{\ell}&=&e_{4}\,T_{I}^{n+1}-\sigma\,V^{n}_{I}.\label{chap7_e2.7}
\end{eqnarray}
Note that the variables $T_{4}^{n}$, $T_{I}^{n}$, and $V_{I}^{n}$ denote the
estimation values of
$T_{4}(n\ell)$, $T_{I}(n\ell)$, and $V_{I}(n\ell)$ respectively. In addition,  $T_{4}^{n}$, $T_{I}^{n}$, and $V_{I}^{n}$
are taken into consideration as a set with three non-negative sequences consistent with the biological nature of the model in practice. In
fact, the parameters of the model should be chosen such that
$T_{4}^{n}$, $T_{I}^{n}$, and $V_{I}^{n}$ are non-negative.
In most cases, the non-negativity is a standard assumption in the
literature of difference equations \cite{chap7_M. Dehghan 2004-2, chap7_M.
Jaberi, chap7_V:Kocic, chap7_K:Kulenovic}, besides the nature of the model
(the IVP system) is another reason. Note also that it is a very
difficult task to investigate \cite{chap7_M. Dehghan 2004-1} the
behavior of solutions of a difference equation which attains
negative and positive values simultaneously. For instance, finding
a forbidden set, the set of initial conditions through which the
related difference equation is undefined, is a difficult task for
this class of difference equations. To observe these types of
complexities, readers are encouraged to study (\cite{chap7_K:Kulenovic}, p.17). However, in this
work we will obtain some results where $\hat{c}>0$ in Equation
(\ref{chap7_e2.5}), that is, the solutions of system
(\ref{chap7_e2.5})$-$(\ref{chap7_e2.7}) deal with negative numbers, see Appendix \ref{sec6}.

Using the following change of variables
\begin{align}
T_{4}^{n}=l\,s\,x_{n}, \quad T_{I}^{n}=({1+l\,\sigma})\,y_{n}/{l\,e_4}, \quad  V_I^n=z_n,
\end{align}
and re-arranging system (\ref{chap7_e2.5})$-$(\ref{chap7_e2.7}), it follows that
\begin{eqnarray}
x_{n+1}&=&\frac{1+x_{n}}{b+c\,z_{n}},\label{chap7_e4.4}\\
y_{n+1}&=&d\,y_{n}+e\,x_{n+1}\,z_{n},\label{chap7_e4.5}\\
z_{n+1}&=&{f\,z_{n}+y_{n+1}},\label{chap7_e4.6}
\end{eqnarray}
where $b=1+l\,\gamma_1>1$, $c=-l\,e_1\in \mathbb{R}$, $0<d=({1+l\,e_3})^{-1}<1$,
$e=l^3 s\,e_4\,e_2\left(({1+l\,e_3})({1+l\,\sigma})\right)^{-1}>0$, and $0<f=({1+l\,\sigma})^{-1}<1$.

%%\begin{eqnarray}
%%T_{4}^{n+1}&=&({\hat{a}+T_{4}^{n}})/({\hat{b}+\hat{c}\,V_{I}^{n}}),\label{chap7_e4.1}\\
%%T_{I}^{n+1}&=&\hat{d}\,T_{I}^{n}+\hat{e}\,T_{4}^{n+1}\,V_{I}^{n},\label{chap7_e4.2}\\
%%V_{I}^{n+1}&=&\hat{f}\,V_{I}^{n}+\hat{g}\,T_{I}^{n+1},\label{chap7_e4.3}
%%\end{eqnarray}
%where $\hat{a}=ls$, $\hat{b}=1+l\,\gamma_1$, $\hat{c}=-l\,e_1$, $\hat{d}={1}/({1+l\,e_3})$, $\hat{e}={l\,e_2}/({1+l\,e_3})$, $\hat{f}={1}/({1+l\,\sigma})$, $\hat{g}={l\,e_4}/({1+l\,\sigma})$. Note that all parameters are nonnegative apart from $\hat{c}$.
%We use the change of variables $T_{4}^{n}=\hat{a}\,x_{n}$,
%$T_{I}^{n}=y_{n}/{\hat{g}}$, and $V_I^n=z_n$ to reduce system
%(\ref{chap7_e4.1})$-$(\ref{chap7_e4.3}) to

Define $\mathcal{R}_d={e}/\left({(b-1)\,(1-d)\,(1-f)}\right)$, the threshold $\mathcal{R}_d$ is the basic reproduction number for discrete model (\ref{chap7_e4.4})$-$(\ref{chap7_e4.6}).
Suppose that $c\neq0$ and $\mathcal{R}_d\neq1$, then it is easy
to show that system (\ref{chap7_e4.4})$-$(\ref{chap7_e4.6}), independent of time step size, has the disease-free  equilibrium point
\begin{equation}\label{chap7_e4.7}
\left(({b-1})^{-1},0,0\right),
\end{equation}
and  the endemic equilibrium point
\begin{equation}\label{chap7_e4.8}
\left(\frac{(1-d)\,(1-f)}{e},\frac{e+(1-b)\,(1-d)\,(1-f)}{c\,(1-d)}, \frac{e+(1-b)\,(1-d)\,(1-f)}
{c\,(1-d)\,(1-f)}\right).
\end{equation}
Note that system (\ref{chap7_e4.4})$-$(\ref{chap7_e4.6}) has only one
equilibrium point (\ref{chap7_e4.7}) where either $c\neq0$ or
$\mathcal{R}_d\neq1$. We also observe that system
(\ref{chap7_e4.4})$-$(\ref{chap7_e4.6}) has the following equilibrium point $({1}/({b-1}),(1-f)\,\psi,\psi)$, where $c=0$ and $\mathcal{R}_d=1$, in which $\psi$ depends on the coefficients of system (\ref{chap7_e4.4})$-$(\ref{chap7_e4.6}) and the
initial value $(x_{0},y_{0},z_{0})$.

%%%%%%%%%%%%%%%%%%%%%%%%%%%%%%%%%%%%%%%%%%%%%%%%%%%%%%%%%%%%%%%%%%%%%%%%%%%%%%
The following theorem provides us the boundedness of solutions which will be trapped into some intervals.
\begin{Thm}[See \cite{chap7_DNJ}, Section 4.1]\label{chap7_t4.1}
Let $b>1$, $c>0$, $0<d<1$, $e>0$ and $0<f<1$. Then either
$(x_{n},y_{n},z_{n})$ converges to an equilibrium point or the
following statements are true.
\begin{eqnarray}
\Gamma \leq \liminf x_n \leq &\limsup x_n &\leq ({b-1})^{-1 },\label{chap7_e4.24}\\
0\leq \liminf y_n\leq&\limsup y_n &\leq \frac{b\,e}{c(1-d)}\left(\frac{1}{b-1}-\Gamma\right),\label{chap7_e4.25}\\
0\leq \liminf z_n\leq&\limsup z_n& \leq
\frac{b\,e}{c(1-d)(1-f)}\left(\frac{1}{b-1}-\Gamma\right).\label{chap7_e4.26}
\end{eqnarray}
where $\Gamma=\min\, (\frac{(b-1)\,(1-d)\,(1-f)}{b\,e},\frac{1}{b-1})$.
\end{Thm}

\subsection{Analysis of disease-free and endemic states }\label{chap7_sec5}
We state some theorems that  will be of help in the subsequent sections and  are derived from the theoretical results of Dehghan et al \cite[Section 4.2--4.3]{chap7_DNJ}. These results present the
local stability of both equilibrium points and the global stability
of (\ref{chap7_e4.7}).
%The Jacobian matrix evaluated at the equilibrium point of system of difference equation (\ref{chap7_e4.4})$-$(\ref{chap7_e4.6})
%$(\bar{x},\bar{y},\bar{z})$ is given by
%$$J(\bar{x},\bar{y},\bar{z})=\left[
%         \begin{array}{lccrr}
%              \dfrac{1}{b+c\,\bar{z}} && 0 && -\dfrac{c\,(1+\bar{x})}{{(b+c\,\bar{z})}^{2}} \\
%             \dfrac{e\,\bar{z}}{b+c\,\bar{z}} && d && \dfrac{b\,e\,(1+\bar{x})}{{(b+c\,\bar{z})}^{2}}\\
%              \dfrac{e\,\bar{z}}{b+c\,\bar{z}} && d &&
%              f+\dfrac{b\,e\,(1+\bar{x})}{{(b+c\,\bar{z})}^{2}}
%          \end{array}\right],$$
%and the characteristic polynomial of $J(\bar{x},\bar{y},\bar{z})$
%is
%\begin{equation}\label{chap7_e4.41}
%P(\lambda)=\lambda^{3}-(d+f+u+b\,e\,u^{2}\,v)\lambda^{2}+(d\,u+f\,u+d\,f+e\,u^{2}v)\lambda-d\,f\,u,
%\end{equation}
%where $u=({b+c\,\bar{z}})^{-1}$ and $v=1+\bar{x}$.
%
%The characteristic polynomial at the equilibrium point (\ref{chap7_e4.7})
%is
%%\begin{equation}\label{chap7_e4.42}
%%\left(\lambda-\frac{1}{b}\right)\left[\lambda^{2}-\left(\frac{e}{b-1}+d+f\right)\lambda+d\,f\right].
%%\end{equation}
%Note that the characteristic polynomial at the equilibrium point
%(\ref{chap7_e4.9}) is given by (\ref{chap7_e4.42}), where
%\[
%\frac{e}{b-1}+d+f=1+d\,f,
%\]
%since   $\mathcal{R}_d=1$. Therefore,
%three eigenvalues of the Jacobian matrix at the equilibria
%(\ref{chap7_e4.7}) and (\ref{chap7_e4.9}) are $\{1, \frac{1}{b}, df\}$
%where  $\mathcal{R}_d=1$. We also observe that the characteristic polynomial (\ref{chap7_e4.42}) has three real positive
%roots and their product is smaller than $1$.

\begin{Thm}[See \cite{chap7_DNJ}, Theorem 4.4]\label{chap7_t4.4}
 Consider system (\ref{chap7_e4.4})$-$(\ref{chap7_e4.6}) with $b>1$, $0<d<1$,
$e>0$, and $0<f<1$. Then the following statements are true.
\begin{itemize}
\item[(i)]
If $\mathcal{R}_d<1$ holds, then the equilibrium point
(\ref{chap7_e4.7}) is asymptotic stable.
\item[(ii)]
If $\mathcal{R}_d>1$ holds, then the equilibrium point
(\ref{chap7_e4.7}) is unstable.
\item[(iii)]
If $\mathcal{R}_d=1$ holds, then the roots of the
characteristic polynomial at the equilibrium points
$(\frac{1}{b-1},0,0)$ and $(\frac{1}{b-1},(1-f)\psi,\psi)$ are
$\{1, \frac{1}{b}, d\,f\}$.
\item[(iv)]
The equilibrium point $(\frac{1}{b-1},0,0)$ is not a repeller.
\end{itemize}
\end{Thm}

\begin{Thm}[See \cite{chap7_DNJ}, Theorem 4.7]\label{chap7_t4.7}
Consider system (\ref{chap7_e4.4})$-$(\ref{chap7_e4.6}) with $b>1$, $c>0$,
$0<d<1$, $e>0 $, and $0<f<1$. If the reproduction number $\mathcal{R}_d$ is less than unity, then
the equilibrium point $(\frac{1}{b-1},0,0)$ is globally
asymptotically stable. Therefore the disease can be eradicated from the host.
\end{Thm}

\begin{Thm}[See \cite{chap7_DNJ}, Theorem 4.5]\label{chap7_t4.5}
Consider system (\ref{chap7_e4.4})$-$(\ref{chap7_e4.6}) with $b>1$, $0<d<1$,
$e>0$, and $0<f<1$. Then the following statements are true.
\begin{itemize}
\item[(i)]
If $\mathcal{R}_d>1$ holds, then the equilibrium point
(\ref{chap7_e4.8}) is asymptotic stable. In addition, the system has no
equilibrium point in $\mathbb{R}_{>0}^{3}(=
\left\{(x,y,z)|~x>0,~y>0,~z>0 \right\})$ when $c\leq0$.
\item[(ii)]
If $\mathcal{R}_d<1$ holds, then the equilibrium point
(\ref{chap7_e4.8}) is unstable. In addition, the system has no
equilibrium point in $\mathbb{R}_{>0}^{3}$ when $c\geq0$.
 \item[(iii)]The equilibrium point (\ref{chap7_e4.8}) is not a repeller.
\end{itemize}
\end{Thm}

Theorems \ref{chap7_t4.4} and \ref{chap7_t4.5} show that a bifurcation occurs
when  $\mathcal{R}_d=1$ (See \cite[Corollary 4.6]{chap7_DNJ}).  On the other hand, the disease will persist if $\mathcal{R}_d$ exceeds
unity, where a stable endemic equilibrium exists. This phenomenon, where the disease-free equilibrium
loses its stability and a stable endemic equilibrium appears as $\mathcal{R}_d$ increases through one, is
known as {\em forward bifurcation} \cite{chap7_Castillo1,chap7_Castillo2,chap7_Castillo3,chap7_Dushoff}. This pattern was first noted by Kermack and McKendrick \cite{chap7_Kermack}, and has been observed in the vast majority of disease transmission models in the literature
ever since (see \cite{chap7_Dushoff,chap7_Hethcote1,chap7_Hethcote2,chap7_Lajmanovich} and the references therein). For models that exhibit forward bifurcation,
the requirement $\mathcal{R}_d<1$ is necessary and sufficient for disease elimination. Furthermore, these theorems
express that the equilibria (\ref{chap7_e4.7}) and (\ref{chap7_e4.8}) are not
asymptotically stable simultaneously.
%Therefore, we obtain the following result.
%\begin{Cor}[See \cite{chap7_DNJ}, Corollary 4.6.]\label{chap7_t4.6}
%Consider system of difference equations (\ref{chap7_e4.4})$-$(\ref{chap7_e4.6})
%with $b>1$, $0<d<1$, $e>0$, and $0<f<1$. If $\mathcal{R}_d=1$, then a forward bifurcation occurs. In addition, the bistability phenomenon does not occur.
%\end{Cor}

Now we give the main question of this work which was given as an Open Problem in  \cite[Section 4.3]{chap7_DNJ}. Hopefully we then show that, based on the parameter set requirement, the conjecture is true.\\

{\bf Open Problem :} Show that the positive equilibrium point (\ref{chap7_e4.8}) is globally asymptomatically stable if $\mathcal{R}_d>1$  along with the following conditions imposed to the paramors: $$b>1, \,\,c>0, \,\,0<d<1, \,\,e>0, \mbox{ and } 0<f<1.$$\\

In the first place, we show that the this problem holds under the assumption when $e>(b-1)(1-d)$.
But before demonstrating this statement, we need to state the following theorem as things develop.

\begin{Thm}\label{chap7_t3.n}
Let $[a_1,b_1]$, $[a_2,b_2]$, and  $[a_3,b_3]$ be intervals of real numbers, and let
\begin{equation*}
f_i:[a_1,b_1]\times[a_2,b_2]\times[a_3,b_3]\rightarrow[a_i,b_i],\qquad \mbox{for }~~ i\in\{1,2,3\}
\end{equation*}
be continuous functions. Consider the system of difference equations
\begin{equation}\label{chap7_e11}
\begin{array}{llll}
    x_{n+1}=f_1(x_n,,z_n),\\
    y_{n+1}=f_2(x_{n+1},y_n,z_n), & & n=0,1,2,...\\
    z_{n+1}=f_3(y_{n+1},z_n),
\end{array}
\end{equation}
with initial conditions $(x_0,y_0,z_0)\in [a_1,b_1]\times[a_2,b_2]\times[a_3,b_3]$. Suppose that the following statements are true.
\begin{enumerate}
    \item $f_1(x,z)$ is non-increasing in $z$ and is non-decreasing in
    $x$.

    \item $f_2(x,y,z)$ is non-decreasing in $x$, $y$, and  $z$.

    \item $f_3(y,z)$ is non-decreasing in  $y$ and  $z$.

    \item If $(m_i,M_i)\in[a_i,b_i]^2$, for $i\in\{1,2,3\}$, is a
    solution of the system of equations
    \begin{eqnarray}\label{chap7_e4.30}
    m_i=f_i(m_1,m_2,m_3) & \mbox{and} &  M_i=f_i(M_1,M_2,M_3) \qquad \mbox{for }~~ i\in\{1,2,3\}
    \end{eqnarray}
    then $m_1=M_1$, $m_2=M_2$, and  $m_3=M_3$.

    Then there exists exactly
    one equilibrium point $(\bar{x},\bar{y},\bar{z})$
    of system (\ref{chap7_e11}) and every solution of system (\ref{chap7_e11})
    converges to $(\bar{x},\bar{y},\bar{z})$.
\end{enumerate}
\end{Thm}

\begin{proof}
It is obtained in the same argument as  \cite{chap7_jaberi1,chap7_jaberi2}.
\end{proof}

We showed that  solutions of system (\ref{chap7_e4.4})$-$(\ref{chap7_e4.6}) are
bounded from above in (\ref{chap7_e4.24})$-$(\ref{chap7_e4.26}) (see Theorem
\ref{chap7_t4.1}). We now demonstrate that solutions are bounded from below, namely
$(y_n,z_n)$ are away from $(0,0)$ (this was already shown for
$x_n$).

\begin{Thm}\label{chap7_t4.5.n}
Assume that $e>(b-1)(1-d)$  and  $(y_0,z_0)\neq(0,0)$ hold. Then
every solution of system (\ref{chap7_e4.4})$-$(\ref{chap7_e4.6}) is bounded from below, i.e. $\liminf\,(x_n,y_n,z_n)>(0,0,0)$.
\end{Thm}
\begin{proof}
It was already shown in (\ref{chap7_e4.24}) that  $x_n>0$.
One can use equations
(\ref{chap7_e4.4})$-$(\ref{chap7_e4.6}) to show simply that if $\limsup_{n\rightarrow\infty} y_n$ or $\liminf_{n\rightarrow\infty} z_n$ equals
zero, then $\{y_n\}_{n=0}^{\infty}$ and $\{z_n\}_{n=0}^{\infty}$
converge to zero as well. Now  assume on the contrary that one of the sequences $y_n$ or $z_n$ converges to zero,
then the other one converges to zero as well. So
$(x_n,y_n,z_n)\rightarrow(\frac{1}{b-1},0,0)$ with $y_n>0$ and $z_n>0$ for all $n\in\mathbb{N}$. Let $\epsilon>0$ satisfy
\begin{eqnarray}\label{chap7_e4.5.1.1}
(b-1)(d-1)+e-e(b-1)\epsilon\ge0,
\end{eqnarray}
then there is $N$ such that $x_n$ is very close to $\frac{1}{b-1}$ and for $n\ge N$
\[
\frac{1}{b-1}-\epsilon<x_n<\frac{1}{b-1}+\epsilon.
\]
As a result, for $n\ge N$,
\[
y_{n+1} > d y_n + e \left(\frac{1}{b-1}-\epsilon\right)z_n.
\]
Using the fact that $z_n \ge y_n$, we have
\[
y_{n+1} > d y_n + e (\frac{1}{b-1}-\epsilon)y_n=
\left(d+e(\frac{1}{b-1}-\epsilon)\right)y_n.
\]
In view of (\ref{chap7_e4.5.1.1}), we have $ y_{n+1} > y_n> ...> y_N$, which is a contradiction. So the desired result follows.
\end{proof}

Now we present one of the main theorems of this work which provide partially an answer to
 the Open Problem given in \cite{chap7_DNJ}  when $e>(b-1)(1-d)$.

\begin{Thm}\label{chap7_t4.8}
Assume that $e>(b-1)(1-d)$ holds. Then the
equilibrium point (\ref{chap7_e4.8}) of system (\ref{chap7_e4.4})$-$(\ref{chap7_e4.6}) is globally asymptotically stable.
\end{Thm}
\begin{proof}
It was shown in Theorem \ref{chap7_t4.5} that the
equilibrium point (\ref{chap7_e4.8}) of system (\ref{chap7_e4.4})$-$(\ref{chap7_e4.6}) is locally asymptotically stable.
It is just required to demonstrate that every solution of system (\ref{chap7_e4.4})$-$(\ref{chap7_e4.6}) converges to the
equilibrium point (\ref{chap7_e4.8}). We use the result of Theorems \ref{chap7_t3.n} and \ref{chap7_t4.5.n}.
It follows by Theorem \ref{chap7_t4.5.n} that there are $(m_x,M_x)$, $(m_y,M_y)$, and $(m_z,M_z)$, a solution of (\ref{chap7_e4.30}), such that
\begin{eqnarray}
0<m_x = \liminf x_n \le x_n \le \limsup x_n = M_x, \\
0<m_y = \liminf y_n \le y_n \le \limsup y_n = M_y, \\
0<m_z = \liminf z_n \le z_n \le \limsup z_n = M_z.
\end{eqnarray}
Now if we show that $m_x=M_x$, $m_y=M_y$, and $m_z=M_z$, then we prove the part (4) of Theorem \ref{chap7_t3.n} and other parts of this theorem are clearly true. To this end, we substitute $(m_x,M_x)$, $(m_y,M_y)$, and $(m_z,M_z)$ in
system (\ref{chap7_e4.4})$-$(\ref{chap7_e4.6}), it then follows that
\begin{eqnarray}
m_x=\frac{1+m_x}{b+c\,M_z} & \mbox{and} & M_x=\frac{1+M_x}{b+c\,m_z},\label{chap7_e5.9}\\
m_y=d\,m_y+e\,m_x\,m_z & \mbox{and} & M_y=d\,M_y+e\,M_x\,M_z,\label{chap7_e5.10}\\
m_z=f\,m_z + m_y \quad &\mbox{and} & M_z=f\,M_z + M_y .\label{chap7_e5.11}
\end{eqnarray}
First isolating $m_x$ and $M_x$ in (\ref{chap7_e5.9}) give
\begin{eqnarray*}
m_x=\frac{1}{b-1+c\,M_z} & \mbox{and} & M_x=\frac{1}{b-1+c\,m_z},
\end{eqnarray*}
we then replace $m_x$ and $M_x$ in (\ref{chap7_e5.10}). Likewise solving (\ref{chap7_e5.10}) for $m_y$ and $M_y$  gives
\begin{eqnarray*}
m_y=\frac{e\,m_z}{(1-d)(b-1+c\,M_z)} & \mbox{and} & M_y=\frac{e\,M_z}{(1-d)(b-1+c\,m_z)},
\end{eqnarray*}
and  we are now able to generate $m_z$ and $M_z$ in (\ref{chap7_e5.11}),
\begin{eqnarray*}
m_z=\frac{e\,m_z}{(1-f)(1-d)(b-1+c\,M_z)} & \mbox{and} & M_z=\frac{e\,M_z}{(1-f)(1-d)(b-1+c\,m_z)},
\end{eqnarray*}
which is equivalent to
\begin{eqnarray*}
1=\frac{e}{(1-f)(1-d)(b-1+c\,M_z)} & \mbox{and} & 1=\frac{e}{(1-f)(1-d)(b-1+c\,m_z)}.
\end{eqnarray*}
This shows that $m_z=M_z$, so are $m_x=M_x$ and $m_y=M_y$. As a result the assumptions of Theorem \ref{chap7_t3.n} hold, then there exists exactly
    one equilibrium point (\ref{chap7_e4.8})
    of system  (\ref{chap7_e4.4})$-$(\ref{chap7_e4.6}) and every solution of system  (\ref{chap7_e4.4})$-$(\ref{chap7_e4.6})
    converges to (\ref{chap7_e4.8}).
\end{proof}

\begin{Thm}\label{chap7_t4.6.6}
Consider system of difference equation (\ref{chap7_e4.4})$-$(\ref{chap7_e4.6})
with $b>1$, $c>0$, $0<d<1$, $e>0$, and $0<f<1$. Assume that
$(x_0,y_0,z_0)\in \mathbb{R}^3_{+}$, then
$\{(x_n,y_n,z_n)\}_{n=0}^{\infty}$ converges to (\ref{chap7_e4.7}) or
(\ref{chap7_e4.8}).
\end{Thm}
\begin{proof}
Using Theorems \ref{chap7_t4.7} and \ref{chap7_t4.8}, it is sufficient to show  the argument hold only for
\[
(b-1)(1-d)(1-f) <e <(b-1)(1-d).
\]
The same reasoning as in the proof of Theorem \ref{chap7_t4.8} produces system (\ref{chap7_e5.9})$-$(\ref{chap7_e5.11}).
If $m_y$ or $m_z$ is zero, say $m_y$, then  by (\ref{chap7_e4.5}) and the fact that  $m_x$ is bounded and greater
than zero it follows that  $m_z$ is zero. As a result by (\ref{chap7_e5.9}), we deduce that $m_x=M_x=1/(b-1)$, hence $M_y=M_z=0$.
So the solution converges to (\ref{chap7_e4.7}). If both the lower bounds are strictly greater than zero, the result is given in Theorem \ref{chap7_t4.8}.
\end{proof}

\begin{Thm}\label{chap7_t4.5.n.n}
Assume that $(b-1)(1-d)(1-f)<e<(b-1)(1-d)$  and  $(y_0,z_0)\neq(0,0)$ hold. Then
every solution of system (\ref{chap7_e4.4})$-$(\ref{chap7_e4.6}) is bounded from below.
\end{Thm}\begin{proof}
Assume on the contrary that one of $y_n$ or $z_n$ converges to zero,
the other one converges to zero too, and then
$x_n\rightarrow\frac{1}{b-1}$. Since $(b-1)(1-d)(1-f)<e$, choose $\epsilon>0$ and $l>0$ such that
\begin{eqnarray}\label{chap7_e4.5.1.1.2}
(1-b)(1-d)(1-f)+e(1-d^{l+1})-e(b-1)\epsilon\ge0,
\end{eqnarray}
then there is $N$ such that $x_n$ is very close to $\frac{1}{b-1}$ and for $n\ge N$
\[
\frac{1}{b-1}-\epsilon<x_n<\frac{1}{b-1}+\epsilon.
\]
As a result, for $n\ge N$
\[
y_{n+1} > d y_n + e \,K\,z_n.
\]
where $K=(\frac{1}{b-1}-\epsilon)$. Let $z_{i,l}=\min\{z_{i},z_{i-1},...,z_{i-l}\}$, then by (\ref{chap7_e4.5.1.1.2}) it follows that
\begin{eqnarray*}
z_{n+1} &>& f z_n  + e \,K\,z_n+ d y_n> f z_n  + e \,Kz_{n-l}+e \,K\,dz_{n-l}+...+e \,Kd^{l}z_{n-l}+ d^ l y_{n-l}\\
        &>&(f + e \,K+e \,K\,d+...+e \,Kd^{l})z_{n,l}+ d^l y_{n-l}\\
        &>& \left(f+e\,K\left(\frac{1-d^{l+1}}{1-d}\right)\right)z_{n,l}.
\end{eqnarray*}
Then by induction on $n$, it follows that
\begin{eqnarray*}
z_{n+T} >...>z_{n+2,l} >z_{n+1,l}>z_{n,l}
\end{eqnarray*}
which is a contradiction. So the desired result obtains.
\end{proof}

The following theorem is the main  result of this research which answers
to the Open Problem given in \cite{chap7_DNJ}.
\begin{Thm}\label{t13}
Assume that $e>(b-1)(1-d)(1-f)$ holds. Then the
equilibrium point (\ref{chap7_e4.8}) of system (\ref{chap7_e4.4})$-$(\ref{chap7_e4.6}) is globally asymptotically stable.
\end{Thm}
\begin{proof}
It follows from  Theorem \ref{chap7_t4.5.n.n} and the argument of Theorem \ref{chap7_t4.8}.
\end{proof}

\section{Inference in continuous HIV infection model with correlated
dynamic discrete model}\label{sec5}
%Dynamical comparison of continuous and discrete models

This section develops a reliable comparison for the original continuous HIV infection model, system
\eqref{chap7_e2.1}--\eqref{chap7_e2.3}, with estimation of correlated dynamic discrete model, system
(\ref{chap7_e4.4})$-$(\ref{chap7_e4.6}). Inference in dynamic discrete model involves computing the solutions to a high-dimensional
system and analyzing theoretically the long-term behavior of solutions. The parameter values used for the numerical scheme, system (\ref{chap7_e4.4})$-$(\ref{chap7_e4.6}), are derived technically by employing the change of variables in Section \ref{sec4}, given as follows,
$$ b=1+l\gamma_{1}, \,\, c=-le_{1}, \,\,d=\frac{1}{1+l\,e_{3}},
\,\,e=\frac{{l}^3\,s\,e_{2}e_{4}}{(1+l\,e_{3})(1+l\,\sigma)}, \mbox{ and }f=\frac{1}{1+l\,\sigma}$$
from which
%and the definition of the model parameters in Table \ref{chap7_table0},
it follows that  $b>1$, $c\in\mathbb{R}$, $0<d<1$, $e>0$, and $ 0<f<1$
when $l>0$. The parameter values for the continuous HIV infection model, system (\ref{chap7_eq2}), were correspondingly  chosen by \eqref{chap7_e2.4} in Section \ref{sec2} and the change of variables \eqref{chap7_eq1}  so that all the possible observed steady states $(\bar x,\bar y,\bar z)$ were identified clearly by
 $$\tau=\gamma_1, \,\,\rho=\frac{e_2\,e_4\,s}{e_3\,\gamma_1}, \,\,\zeta=e_3, \mbox{ and } \eta=\sigma.$$

The chain convergence is checked by comparing the estimation results for several different parameter values in Theorems \ref{chap7_t4.4} and
\ref{chap7_t4.5} with the results of \cite[Theorem 2.1]{chap7_M. Dehghan 2007} or \cite[Section 4]{chap7_A.B. Gumel}. Surprisingly, we observe that the local behavior of solutions of the continuous HIV infection model is qualitatively similar to that of numerical scheme for discrete model corresponding to the same parameter set, since the
conditions  $\mathcal{R}_d=\frac{e}{(b-1)\,(1-d)\,(1-f)}>1$, $\mathcal{R}_d=1$, or $\mathcal{R}_d<1$
%\[
%\mathcal{R}_d=\frac{e}{(b-1)\,(1-d)\,(1-f)}< or = or >
%\]
%$(b-1)\,(1-d)\,(1-f)<e$, $(b-1)\,(1-d)\,(1-f)=e$, and
%$(b-1)\,(1-d)\,(1-f)>e$
are equivalent to $\mathcal{R}_c=\frac{s\,e_2\,e_4}{\sigma\,\gamma_1\,e_3}>1$, $\mathcal{R}_c=1$, or $\mathcal{R}_c<1$
%$\sigma\,\gamma_1\,e_3<s\,e_2\,e_4$, $\sigma\,\gamma_1\,e_3=s\,e_2\,e_4$, and
%$\sigma\,\gamma_1\,e_3>s\,e_2\,e_4$,
respectively. In addition,
the conditions $d<b$, $d=b$, and $d>b$ are equivalent to
$\sigma\,\gamma_1\,e_3<s\,e_2\,e_4$, $\sigma\,\gamma_1\,e_3=s\,e_2\,e_4$, and
$\sigma\,\gamma_1\,e_3>s\,e_2\,e_4$, respectively. Therefore, Theorems
\ref{chap7_th2}, \ref{chap7_t4.7}, and  \ref{chap7_t4.6.6}, and Corollary 4.3 in \cite{chap7_DNJ} show
that global dynamics of these models are 	considerably equivalent for
$e_1 < 0$ which is related to $c>0$. The estimation results for the correlated dynamic discrete model are only slightly affected by the fixed time step size $l$, see Subsection \ref{subsec1}. Such a case scenario shows that the numerical scheme \eqref{e1.1} is an efficient algorithm in conjunction with $e_1 < 0$
and suitable for use to solve the dynamic continuous HIV infection model.

The following results are motivated by the stage estimation procedure in which
the key parameter $c$ is taken to be zero, that is, $c=0$ corresponding to $e_1=0$ for the dynamic continuous model. In fact, these results are provided by Theorem \ref{chap7_th3} in \cite{chap7_M. Dehghan 2007} (proved for the continuous model) and  Theorem 4.2 and Corollary 4.3 in \cite{chap7_DNJ} (proved for the discrete
model). The paper studies theoretically the estimation method's performance in the case when $e_1=0$ or $c=0$. The findings demonstrate that  the estimation accuracy of the proposed numerical method is excellent since similar results for both models are achieved.

However, the main difference appears in the qualitative behaviors of solutions of the models when $e_1>0$. Theoretical results and several experiments on artificial data and parameter values show that the dynamic discrete model conducted by the estimated method and the continuous model with HIV infection can behave quite differently when $e_1>0$. More generally, the experiments demonstrate that  system (\ref{chap7_eq2}) can attain bounded or unbounded
solutions \cite{chap7_M. Dehghan 2007} while the qualitative behavior of system (\ref{chap7_e4.4})$-$(\ref{chap7_e4.6}) with different initial values can lead to
dispersion and serious misspecification errors referred to as forbidden sets, see Appendix A. The proposed theoretical framework for system
(\ref{chap7_eq2})  is flexible and leaves room for experimentation and abstract results. Experiments with the numerical estimation algorithm
for solving system (\ref{chap7_e4.4})$-$(\ref{chap7_e4.6}) requires a discovery of modifications and development of the numerical method in our
future work that will hopefully provide a very efficient alternative to iterating the approximation \eqref{e1.1} for negative parameter $c$. The
most important reason for this difference seems to be, a great deal, based on
%first-order
forward-difference approximation \eqref{e1.1}.

%????we have the same result for  and .
%
%?????
%\begin{Thm}[See ]\label{chap7_t4.2}
%Let $b>1$, $c=0$, $1>d>0$, $e>0$ and $1>f>0$ satisfying
%$\mathcal{R}_d\neq1$. Then either
%\begin{equation}\label{chap7_e4.33}
%(x_n,y_n,z_n)\rightarrow\left(\frac{1}{b-1},\infty,\infty\right)
%\end{equation}
%or
%\begin{equation}\label{chap7_e4.38}
%(x_{n},y_{n},z_{n})\rightarrow\left(\frac{1}{b-1},0,0\right).
%\end{equation}
%\end{Thm}
%
%
%\begin{Cor}[See \cite{chap7_DNJ}]\label{chap7_t4.3}
%Consider system of difference equation (\ref{chap7_e4.4})$-$(\ref{chap7_e4.6})
%with $b>1$, $0<d<1$, $e>0$, $0<f<1$, and  $\mathcal{R}_d=1$.
%\begin{itemize}
%\item[(i)]
%If $c>0$, then (\ref{chap7_e4.38}) is true.
%\item[(ii)]
%If $c=0$, then we have
%\begin{equation*}\label{chap7_e4.40}
%(x_{n},y_{n},z_{n})\rightarrow\left(\frac{1}{b-1},(1-f)\,\psi,\psi\right).
%\end{equation*}
%\end{itemize}
%\end{Cor}

\section{Numerical experiments}\label{sec4.1}
A verification of the algorithm implementation for the above scenario cases is provided in this section. For example, the  effect of different parameter values for time-step and immune response  to the proposed estimation algorithm \eqref{e1.1} are studied to identify an effective treatment for the management of HIV infection administrating multiple anti-HIV preventive drugs in vivo. Multiple parameter sets are used to check the convergence of the estimation procedure in which solving the exact solution of  model (\ref{chap7_eq2}) can be quickly computed. Also several examples show that with different initial values $(x_0,y_0,z_0)\in\mathbb{R}^3_{+}$, solutions of (\ref{chap7_e4.4})$-$(\ref{chap7_e4.6}) attain negative values. However, theoretical results are provided to show that some solutions of system (\ref{chap7_eq2}) remains in $\mathbb{R}^3_{+}$ whenever $(x_0,y_0,z_0)\in\mathbb{R}^3_{+}$ in Appendix A.

\subsection{Effect of time-step, $l$}\label{subsec1}
We will show that the associated numerical estimation method is very robust in terms of numerical stability since the equilibria and the stability conditions are independent of the time step by detailed comparisons. Several numerical simulations were carried out using the numerical method (\ref{e1.1}), regarded as Method A in Table \ref{chap7_table1}, to identify the behavior of solutions of  the model (\ref{chap7_e4.4})$-$(\ref{chap7_e4.6}) with various time steps. An intuitive comparison of the results obtained from the application of the fourth-order Runge-Kutta method (RK4) and  Method A is depicted in Table \ref{chap7_table1} from which it is obvious that numerical scheme (\ref{e1.1}) is more competitive in connection with numerical stability. Extensive numerical simulations performed by the numerical approximation (\ref{e1.1}) show that Method A did not provide solutions with chaotic behavior \cite{chap7_A.B. Gumel}. In fact,  every solution of estimation method converges either to the disease-free state or endemic state. This scenario was given by some numerical examples in  \cite{chap7_A.B. Gumel} without any theoretical proof.

Interestingly the equilibrium point of system
(\ref{chap7_e2.5})$-$(\ref{chap7_e2.7}), say $(\bar T_4, \bar T_I, \bar V_I)=(\hat a\, \bar x, \bar y/\hat g, \bar z)$, is independent of the time step, that is, the solution approaches to the disease-free or endemic equilibrium point of the original model given by the system of ordinary differential equations (\ref{chap7_e2.1})$-$(\ref{chap7_e2.3}) or system (\ref{chap7_eq2}). As a result, regardless of the  implementation of  time step, a unique equilibrium point is eventually obtained. It is also shown in Theorems \ref{chap7_t4.7} and \ref{t13} that the solution obtained by the numerical scheme  approaches to the disease-free or endemic equilibrium point, whether the reproduction number $\mathcal{R}_d$ is less than or in excess of unity. The demonstrations are independent of  variation of $l$ as long as $l>0$. Using the previous discussion and dynamical comparison of Section \ref{sec5}, the solution likewise converges to the equilibrium point for sufficiently small values $l$ ($0<l<<1$) and is, as a result, consistent with the compartmental model (\ref{chap7_e2.1})$-$(\ref{chap7_e2.3}).

%\newcolumntype {Z}{ >{\centering \arraybackslash }X}
\begin{table}[]
\tikzstyle{int}=[rectangle, draw, fill=blue!10,
    text width=22em, text centered, rounded corners, minimum height=10.5em]
\tikzstyle{init} = [pin edge={to-,thick,black}]
{
\tikzstyle{cloud} = [draw=none, rectangle,text centered,rounded corners,text width=\textwidth,fill=gray!20, node distance=3.8cm,
    minimum height=5em]%;
\begin{tikzpicture}[node distance=2.cm,auto,>=latex']
    \node [cloud] (dd)
    {\caption{\footnotesize Effect of time step, $l$, on the convergence of the methods using $c = 3.1$ \cite{chap7_A.B. Gumel}.}\label{chap7_table1}
\centering\footnotesize
\begin{tabularx}{\textwidth}{XXl}
\toprule %
  $l$ &  RK4& Method A  \\[1ex]\hline
  0.001 & Convergence & Convergence\\%\otoprule
  0.01 & Divergence & Convergence\\
  10 & Divergence & Convergence\\
  200 & Divergence & Convergence\\[1ex]
\end{tabularx}
};
\end{tikzpicture}
}
\end{table}

\subsection{Effect of basic reproduction number, $R_d$}

Here we attempt to determine appropriately the dynamics of discrete model by the sensitive quantity threshold, $R_d$, keeping the immune response to a sufficient level. A reliable comparison of the estimation results for the proposed numerical approximation \eqref{e1.1} suggests that a qualitative behavioral change to the model appears through destabilizing the disease-free equilibrium and emergence of a stable endemic equilibrium.
A stable endemic equilibrium branch appears with an unstable disease-free equilibrium (leading to the persistence and existence of HIV within the infected individual) when the related reproduction number is less than one. The attractivity of the model reveals that the disease-free equilibrium is globally asymptotically stable under certain assumptions. Nevertheless the unique endemic equilibrium is globally asymptotically stable in the presence of $60\%$ effective combination therapy. In addition, a comparison between the discrete and continuous models is presented to have a better view in understanding the behavioral dynamics of HIV infection modelling.

%The incidence and the prevalence of some diseases oscillate seasonally in
%a population. This oscillation seems to be caused by seasonal oscillation in
%the contact rate â. For example, the incidence of childhood diseases such as
%measles and rubella increases each year in the winter when children aggregate
%in schools [119, 159, 46].
%damps in an oscillatory manner
%
%The
%appears when the basic reproduction number, $R_d$ exceeds unity,  loses its stability and which leads to  as the basic reproduction number.
%In such a scenario, the classical requirement of the reproduction number being less than unity becomes a necessary and sufficient condition for disease elimination.

The correlated dynamic discrete model may exhibit a threshold behavior of a possibility of sustained oscillations. Oscillations in the model mean fluctuations in the number of uninfected/infected/infectious cells to be expected and the HIV infection still persists but in an
oscillatory manner for certain parameters values.
% with the basic reproduction number above a threshold.
The model (\ref{chap7_e4.4})$-$(\ref{chap7_e4.6}) suggests that the presence of damped oscillations is inevitable to emerge if the threshold distinguished as the basic reproduction number is maintained above a certain quantity during therapy. This includes the possibility that the presence of certain parameter values can destabilize the endemic equilibrium point of the model with constant coefficients and lead to sustain damped oscillations for the reproduction number exceeding a threshold value. Such a change in behavioral dynamics of infection is considerably affected by the insufficient level of CTL response in the immediate surroundings of antiretroviral therapy with relapse. Clearly a higher magnitude of basic reproduction numbers  inferred from both continuous HIV infection model and
dynamic discrete model is related to a comparatively lower threshold level of immune response or ineffective combination of reverse transcriptase and protease inhibitors in vivo.

Numerical simulations,  Figure \ref{chap3_pic3}, indicate that transitory oscillations are observed
for about the first 100 units of time and no sustained oscillations are obtained when the endemic equilibrium point appears to be stable for lower values of $R_d$ close to unity, $R_d\geq1$, but not for significantly large values. As the threshold level of basic reproduction number exceeds unity and is maintained above a certain threshold, the time periods of oscillations to the solutions of the model become shorter and rapid
in conjunction with occurrence of higher amplitudes of oscillations.

Comparison of the two rows in Figure \ref{chap3_pic3} for two different parameter values $c=0.6, \,0.006$ shows that the magnitude of maximal density for the activated CD4$^{+}$ T cells
carrying integrated HIV (infected cells) and density of virus particles
%as infected cells carrying integrated
is overwhelmingly larger in the top array corresponding to  smaller value, $c=0.006$, while the estimation results for the density of uninfected cells are only slightly affected by the same parameter values. Surprisingly, the  qualitative oscillatory behavior of the correlated dynamic discrete model is not substantially perturbed by the implementation of various parameter values for $c$. This result seems to have a simple explanation. The expected long-term solution of the model (\ref{chap7_e4.4})$-$(\ref{chap7_e4.6}) associated with the equilibrium point \eqref{chap7_e4.8} is intuitively sensitive to the implementation of the parameter $c$ and the basic reproduction number. The larger the parameter value $c$, the higher amplitude is observed for the density of infected cells and virus particles, Figure \ref{chap3_pic3}. However the parameter value $c$ produces directly no sizeable effects or possible repercussions to the density of uninfected cells.

\begin{figure}[b]
\begin{center}
\begin{tabular}{lll}
   \includegraphics[scale=.27]{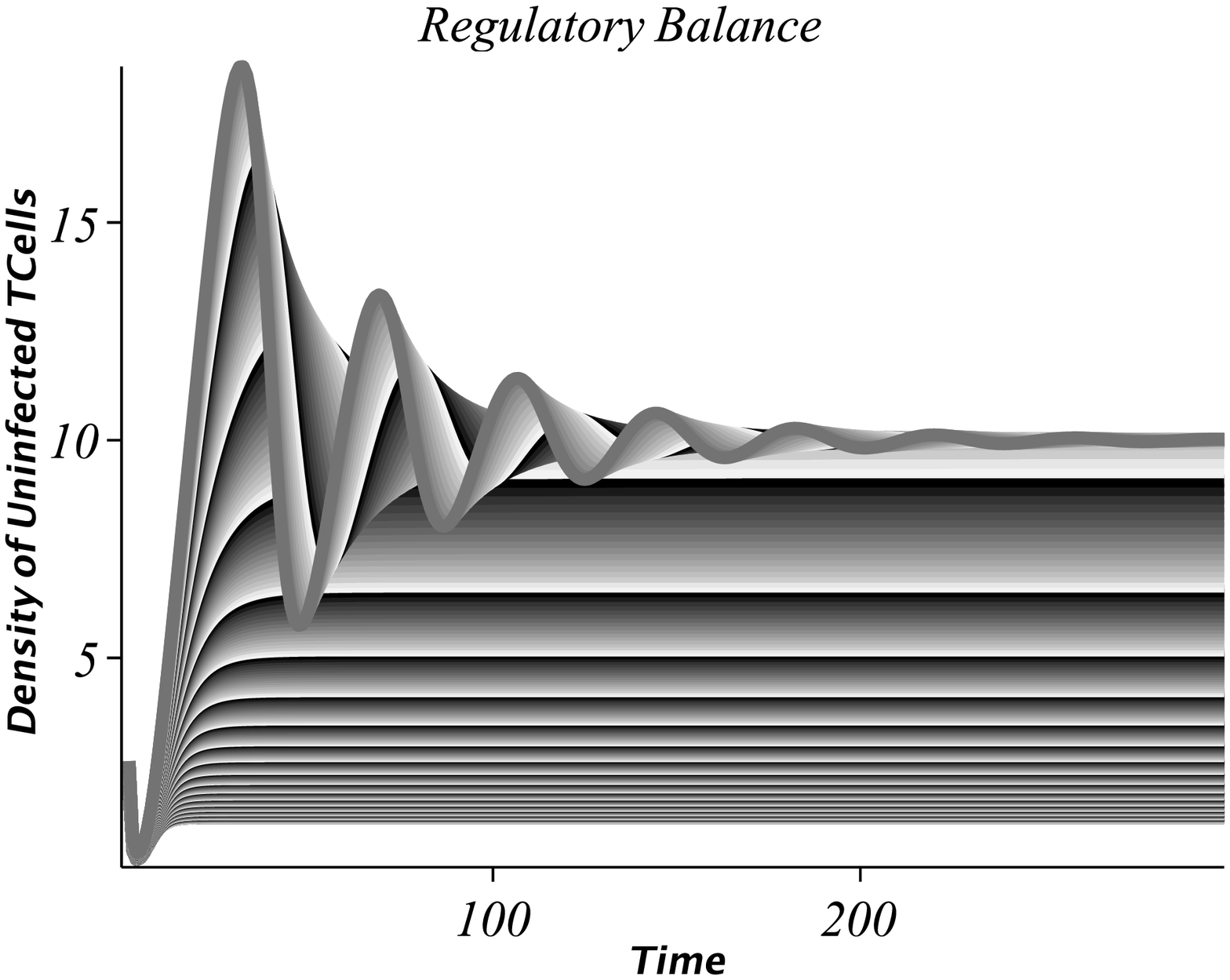}&\includegraphics[scale=.27]{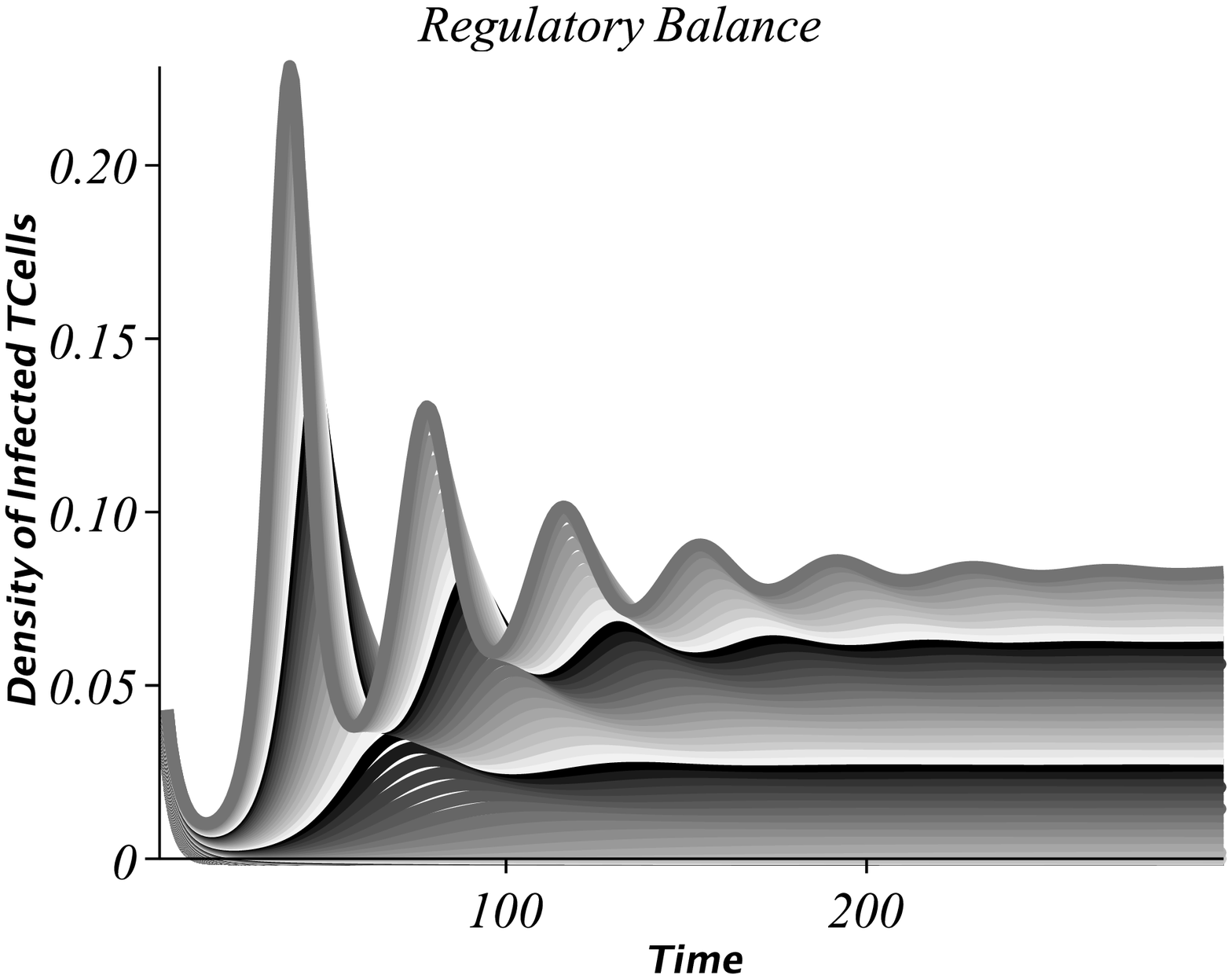}&\includegraphics[scale=.27]{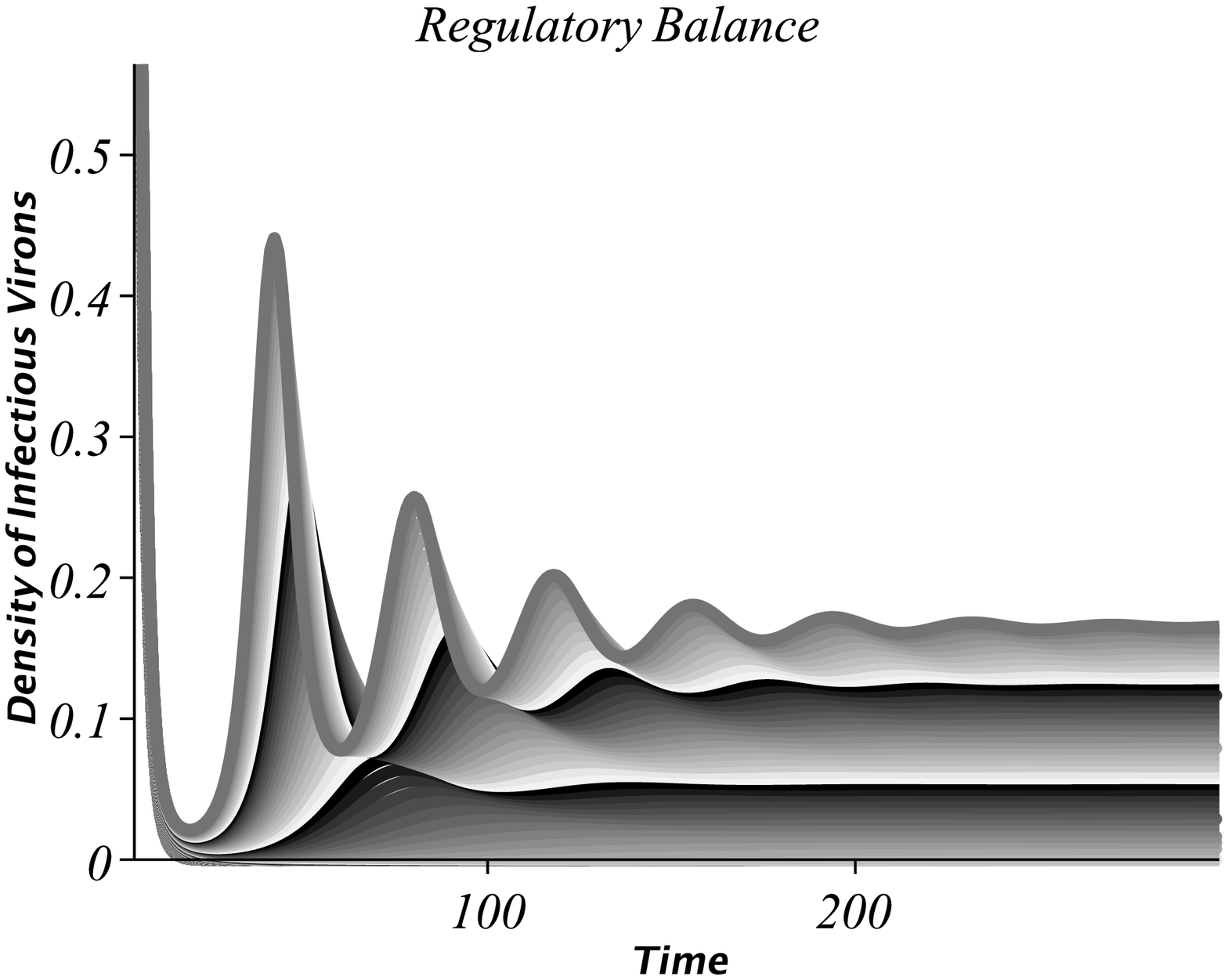}\\
   \includegraphics[scale=.27]{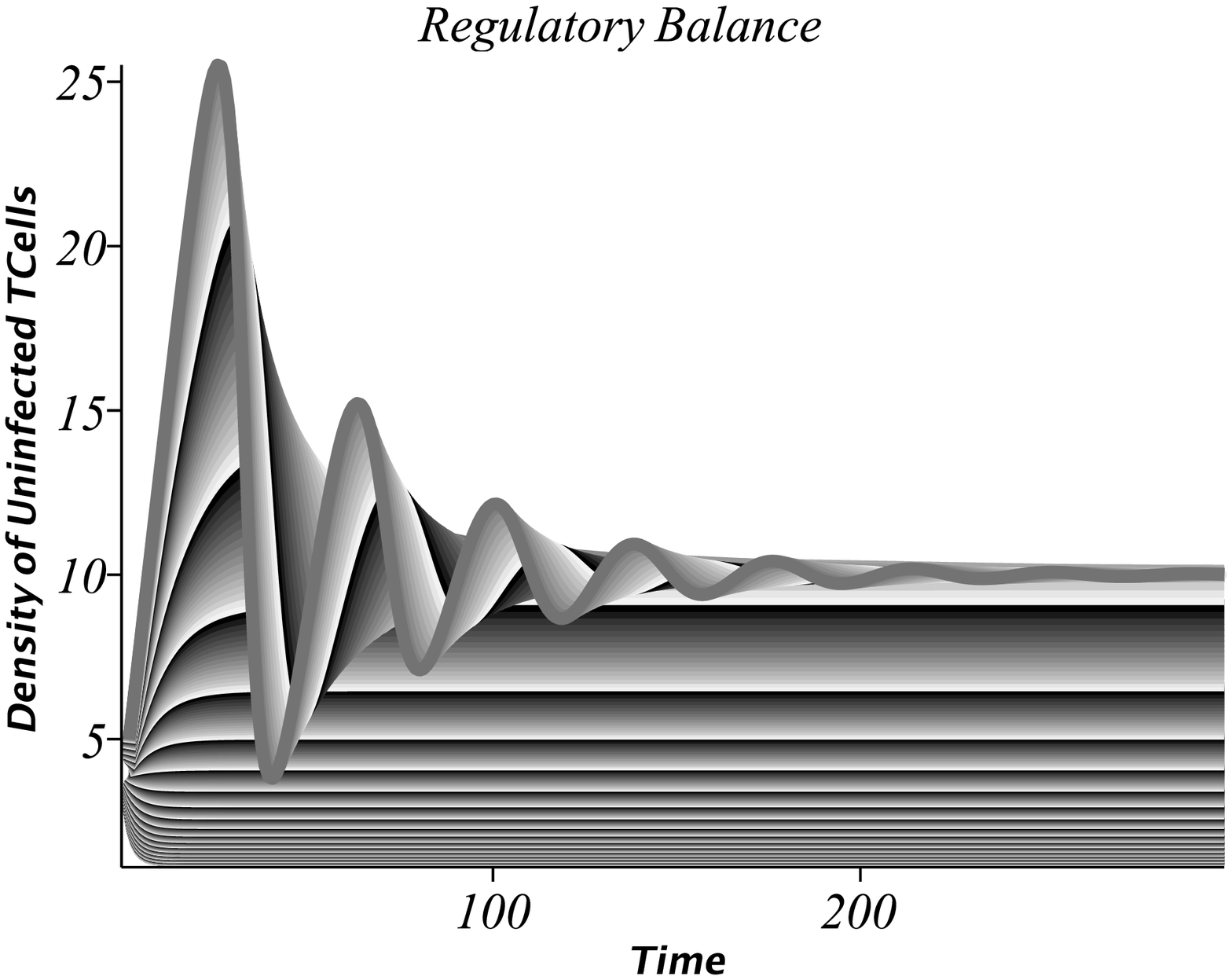}&\includegraphics[scale=.27]{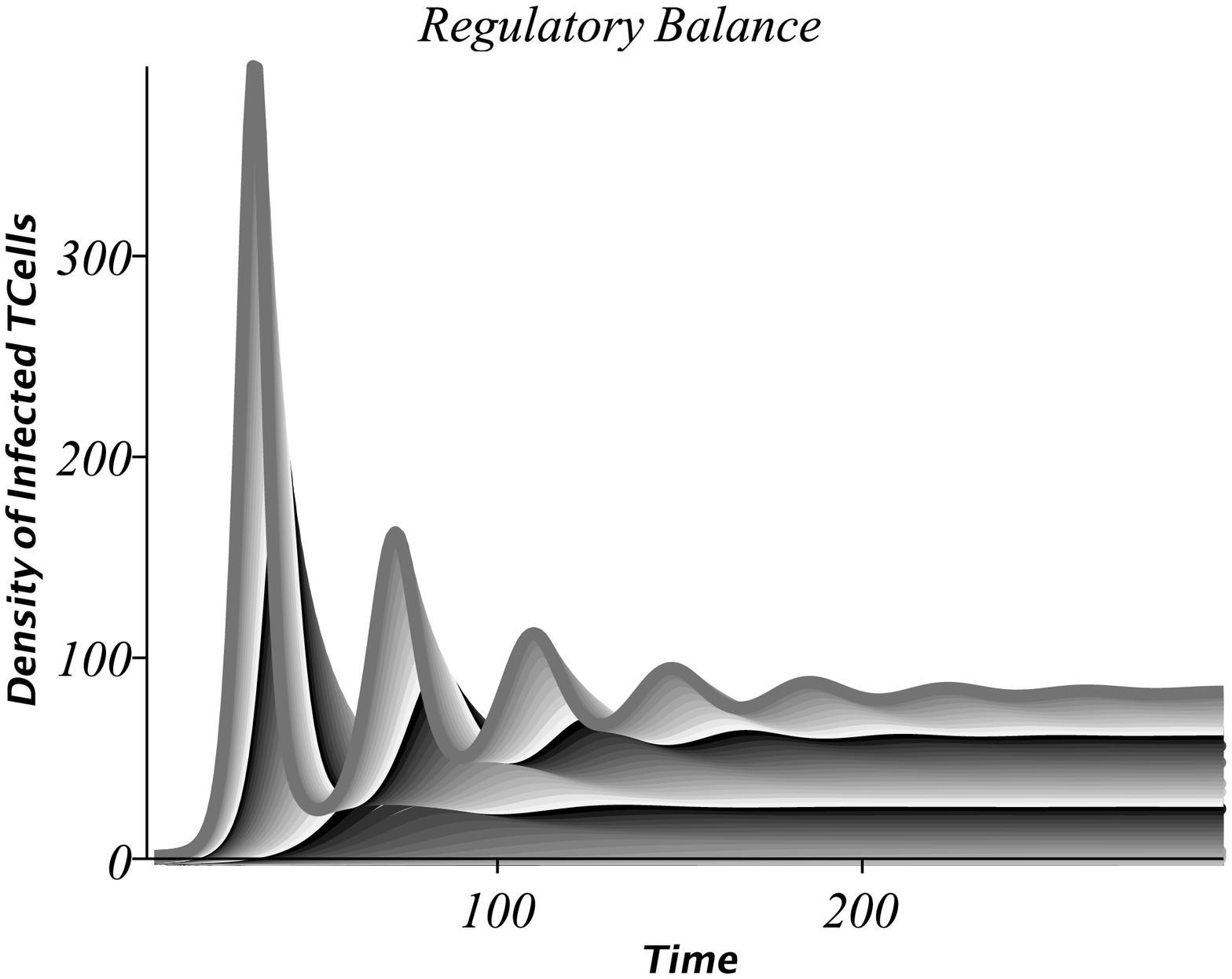}&\includegraphics[scale=.27]{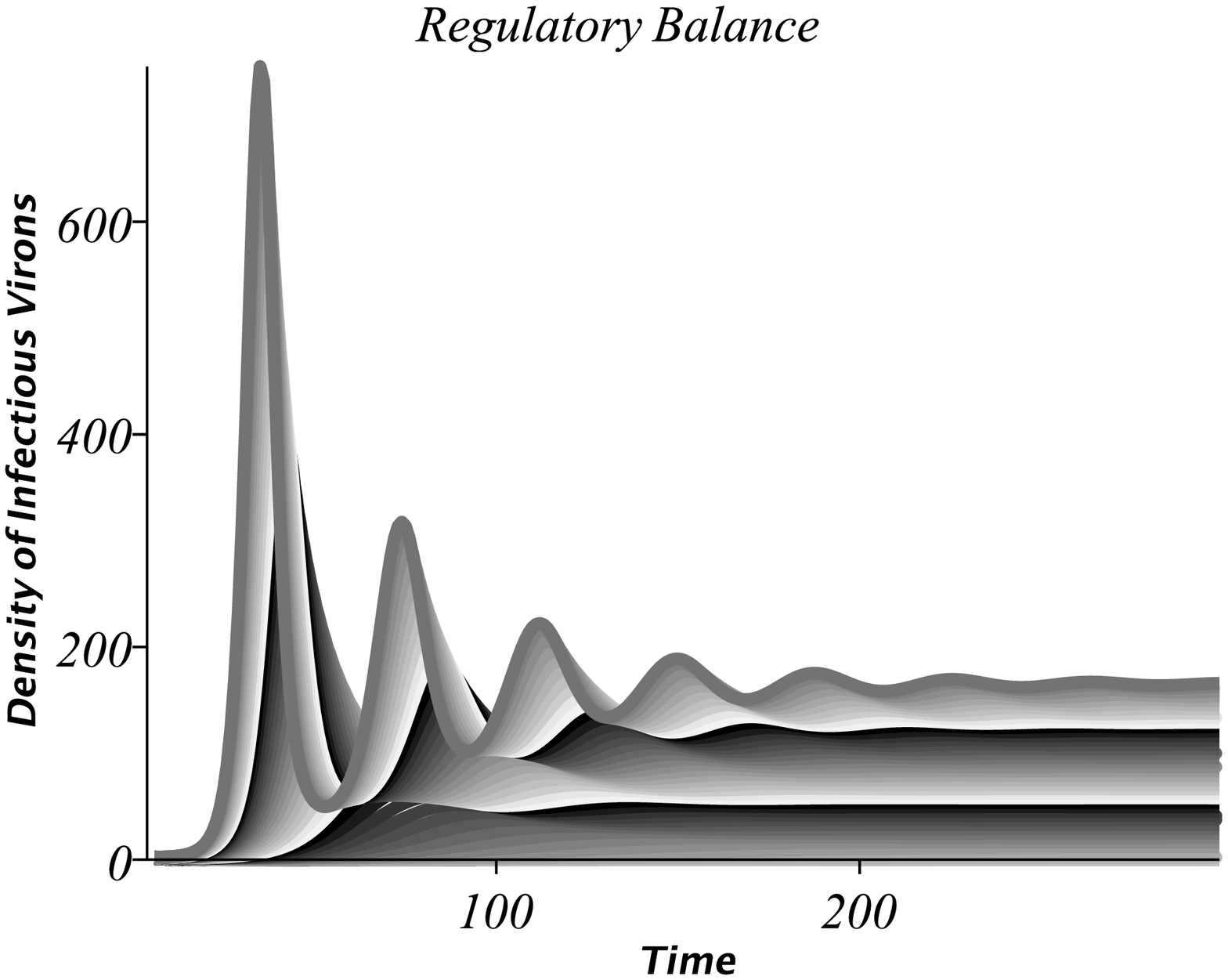}
\end{tabular}
\end{center}
\caption{\footnotesize
Estimated densities of $(\bar x, \bar y, \bar z)$ in conjunction with $(\bar T_4, \bar T_I, \bar V_I)$  are arranged horizontally for two different parameter values; $c=0.6$ is taken into account for the densities located on the top array and $c=0.006$ is considered for the densities on the bottom row.
The accumulated densities corresponding to uninfected/infected/infectious cells are plotted for a range of different threshold levels of the basic reproduction
number, $R_d\in[0.13344,\,100]$. The lower densities related to lower horizontally shaded strips are obtained from the relatively high threshold
level of immune response along with the effectiveness of reverse transcriptase and protease inhibitors, that is, $R_d\leq1$. The innermost strips
with tardily transitory oscillations are illustrated for the threshold level of basic reproduction number exceeding unity, but not for significantly
large values of $R_d$. While the relatively large magnitude of $R_d$  may exhibit a threshold behavioral change  which may result in generating sustained
oscillations of the uninfected/infected/infectious cells.
%
%
%N
%= 1000, the most widespread to N = 100. The dashed lines are fitted normal densities. The dashed lines are prior densities.
%The vertical lines show the actual parameter values. The posterior means of the hazard functions. Panel (a)the data generated from the model with
%the serially correlated unobservables, panel (b)the data generated from the dynamic logit model. The vertical
%axis is for the probability of engine replacement, the horizontal axis is for the mileage interval. The solid line is
%for the model with serially correlated unobservables. The dashed linefor dynamic multinomial logit, the dotted
%linedata hazard.
%Comparison with exact and approximate estimation algorithms. Estimated posterior densities: (a) 1,
%(b) , (c) 1, (d) , (e) 3. The vertical lines show the actual parameter values.The solid line shows the posterior
%for exact estimation procedure, the dashed line  approximate estimation procedure
%A comparison between combination antiretroviral therapy versus monotherapy.
%Based on simulations and theoretical results, an effective treatment strategy for the removal of integrated HIV cells is identified by
%administration of combination therapy. From the Figure $(a)-(c)$ we observe that the survival of viral load is more than $\%99$ for
%individuals in vivo while it can be broadly reduced to $\%80$ by application of both $E_{RT}$ and $E_{PI}$.
}
 \label{chap3_pic3}
\end{figure}

\subsection{Effect of immune response, $k_c$}
The effect of the CTL (cytotoxic T lymphocyte) response was investigated \cite{chap7_A.B. Gumel} by simulating Method A with various values of the CTL response parameter, $k_c$, in the absence of therapy (that is $E_{RT} = E_{PI} = 0$). The steady-state
values of the corresponding viral loads are given in the following table. Another experiment was carried out to
monitor $k_c$ when combination therapy is $60\%$ effective ($E_{RT} = E_{PI} = 0.6$) and the results are tabulated
in Table \ref{chap7_tableww}. It is worth mentioning that such limited efficiency of combination therapy ($60\%$) may
occur due to many reasons, including sub-optimal usage of the regimen, poor compliance, poor
absorption of certain drugs, mutation, etc.
\begin{table}[hbt]
\tikzstyle{cloud} = [draw=none, rectangle,text centered,rounded corners,text width=\textwidth,fill=gray!20, node distance=3.8cm,
    minimum height=5em]%;
\begin{tikzpicture}[node distance=2.cm,auto,>=latex']
    \node [cloud] (dd)
    {\caption{\footnotesize Effect of CTL response, steady-state values of infectious viral load (density of HIV) \cite{chap7_A.B. Gumel}}\label{chap7_tableww}
\begin{center}\footnotesize
\begin{tabularx}{\textwidth}{llXX}
\toprule %
$k_c$ && Effect of CTL response in the absence of effective protease inhibitors where $E_{RT} = 0$   &Effect of CTL response in the presence of $60\%$ combination therapy  \\\hline
0.0 &&834.90&588.11\\
%0.2 &&594.93&413.59\\
%0.4 &&461.63&316.63\\
0.6 &&376.77&254.93\\
%0.8 &&318.04&212.21\\
1.0 &&274.97& 180.89\\
  [1ex]
\end{tabularx}
\end{center}};
\end{tikzpicture}
\end{table}
\subsection{Combination antiretroviral therapy versus monotherapy}
In this experiment, we make use of the mathematical discrete model (\ref{chap7_e4.4})$-$(\ref{chap7_e4.6}) to identify an effective treatment strategy for the management of HIV infection using multiple anti-HIV preventive drugs in vivo. Figure \ref{chap3_pic1} describing the within-host infectious dynamics provide us an intuitive comparison for both combination antiretroviral therapy and monotherapy using the quantity threshold $\mathcal{R}_d=\frac{s\,e_2\,e_4}{\sigma\,\gamma_1\,e_3}$ related to the set of parameter values $s=8.076$, $r=0.03$, $k_{v}=1$,  $\gamma_{1}=0.5$, $\gamma_{2}=0.5$, $N=1000$, $L=0.2$, $\sigma=10$ and multiple sets of variables $(E_{RT},\,E_{PI},\,\alpha)$. For the monotherapy experiment, we assume that only protease inhibitors, e.g. nelfinavir, are administered to the HIV-infected patient \cite[Section 6.3]{chap7_A.B. Gumel}; thus the horizontal axis as a variable is used for the effectiveness values $E=E_{PI}$ with $E_{RT}=0$. Note that a similar experiment can be illustrated by assuming  $E_{RT}$ as variable with $E_{PI}=0$. As the administration of a combination of two or more drugs to the HIV-infected patient is currently a treatment strategy \cite[Section 6.3]{chap7_A.B. Gumel}, another experiment with respect to the combination therapy is conducted with an equivalent administration of both reverse transcriptase and protease inhibitors, that is, we assume that the horizontal axis as a variable represents the effectiveness values $E=E_{PI}=E_{RT}$. In both experiments, vertical axes represent the effect of the CTL (cytotoxic T lymphocyte) response. The black/gray/white shaded areas are used for various values of pre-existing activated CD4$^{+}$ T cells, $0\le\alpha\le1$. The basic reproduction number is in excess of unity in each corresponding area with the stated value $\alpha$ and including the area(s) with lower values of $\alpha$. For instance, the basic reproduction number for $\alpha=0.1$ not only represents the corresponding shaded area but also includes the areas with $\alpha=0,\, 0.01$.

In both scenarios, further computational experiments of the model for an enhanced comparison are performed on the representation of  $60\%$ and $80\%$ effectiveness of therapy with $50\%$ and $75\%$ immune response. In each set of these independent experiments, four distinguishable horizontal and vertical dashed lines are drawn on the shaded areas for this model, see Figure \ref{chap3_pic1}. These experimental results are intriguing. Using both experiment presented in Figure \ref{chap3_pic1}, clearly the chance to suppress the viral load is so narrow for monotherapy, in some cases less than one percent and it may be interpreted as impossible, while combination therapy provide better confidence in removal of integrated HIV cells. For instance, the experimental analysis on the model with administration of only protease inhibitors refereed as monotherapy indicates that the basic reproduction number is greater than one for the ordered pair $(0.6,0.5)$  with a proportion $\alpha\ge0.01$ of pre-exiting activated CD4$^{+}$ T cells, that is, the virus load cannot be entirely inhibited in the presence of $60\%$ effective monotherapy at $50\%$ effective immune system or more. In contrast, the suppression of virus load in the patient is likely to occur for $\alpha\le0.1$ by  implementing the same experiment using $60\%$ effective combination therapy in the presence of more than $90\%$ effective immune system.

Second analysis on the model may be performed in the immediate surroundings of $80\%$ effective monotherapy and combination therapy to a HIV-infected patient. The findings, based on simulations for the correlated dynamic discrete model, suggest that  the combination therapy is greatly efficient in mitigating  the infectious viral load with a possibility of successful suppression for the proportion $\alpha\le0.5$, since in this case $R_d$ will be equal or inferior to one. On the contrary, by administrating only protease inhibitors in a similar experiment, the infectious virus is most likely to persist within an infected individual if $\alpha\ge0.1$. The amount of viral load is not sufficiently reduced to a  certain extent required for $R_d<1$.

Alternatively, to further verify the possibility of successful persistence or inhibition of infectious viral load in a patient,  we conduct this experiment with  a proportion value of pre-exiting activated CD4$^{+}$ T cells close to unity to perturb the dynamics of the HIV model, these experimental results are also intriguing. To control and bring down the volume of viral load for monotherapy scenario, it is required to monitor treatment not less than $99\%$ effective, and in some cases it is impossible to mitigate HIV risk in vivo. However the consequence of administrating two drugs is a success for mitigation and prevention of infectious virus load using antiretroviral treatment that is $90\%$ effective or more if $\alpha$ is close to one.

%is implemented correctly I conducted a similar simulation
%study using the extreme value iid unobservables instead of the serially correlated unobservables.
%The results were analogous to the ones reported in the figures.
%
%effect of administering both RT and
% is studied. Fig. 2 depicts the pro6le of infectious viral
%load against time when RT and protease inhibitors are operating at $60\%$ eJciency, from which it
%is clear that (despite the presence of residual virus at steady state)
%
%$0\le\alpha\le1$  Investigating theoretical properties of this improved procedure, e.g. deriving
%complexity bounds, seems to be of great interest and is a subject of future work. This improvement
%has not been incorporated into the estimation experiments in this paper. However, I employ it in
%Norets (2006) that uses artificial neural networks to approximate the expected value function as a
%function of the parameters and the state variables.
%
%The comparison of the theoretical models are obtained trough the
%methods of concentrated parameters takes place with the experimental results of a beam with the same conditions of the
%theoretical study
%
%the Modal Analysis consists in determinate a experiment of a ensemble of answer functions
%in frequencies and extracting from their with a computer program the modals parameters of the system.

Besides, Table \ref{chap7_table2}  depicts the steady-state values of the viral load with multiple parameter values for the effectiveness of the inhibitors (RT and PIs). In the numerical simulation, identical efficiency levels were conducted for the management of HIV infection in vivo using the mathematical model (\ref{chap7_e4.4})$-$(\ref{chap7_e4.6}) \cite{chap7_A.B. Gumel}.  It is clearly illustrated that as the efficiency level of the drugs is increased, the infectious virus density decreases. With the administration of $99\%$ effective combination therapy (that is, $E_{RT}=E_{PI}=0.99$), however, the two steady states exchange their stability, and the virus particle density converges to zero \cite{chap7_A.B. Gumel}. The endemic equilibrium loses its stability and a stable disease-free equilibrium  appears as the basic reproduction number falls behind unity. In such a scenario,  the associated reproduction number being inferior to unity becomes a necessary and sufficient condition for disease elimination to the HIV-infected patient. This result is, of course, consistent with the theoretical analysis of previous sections.
\begin{figure}[b]
\begin{center}
\begin{tabular}{ll}
   \includegraphics[scale=.3]{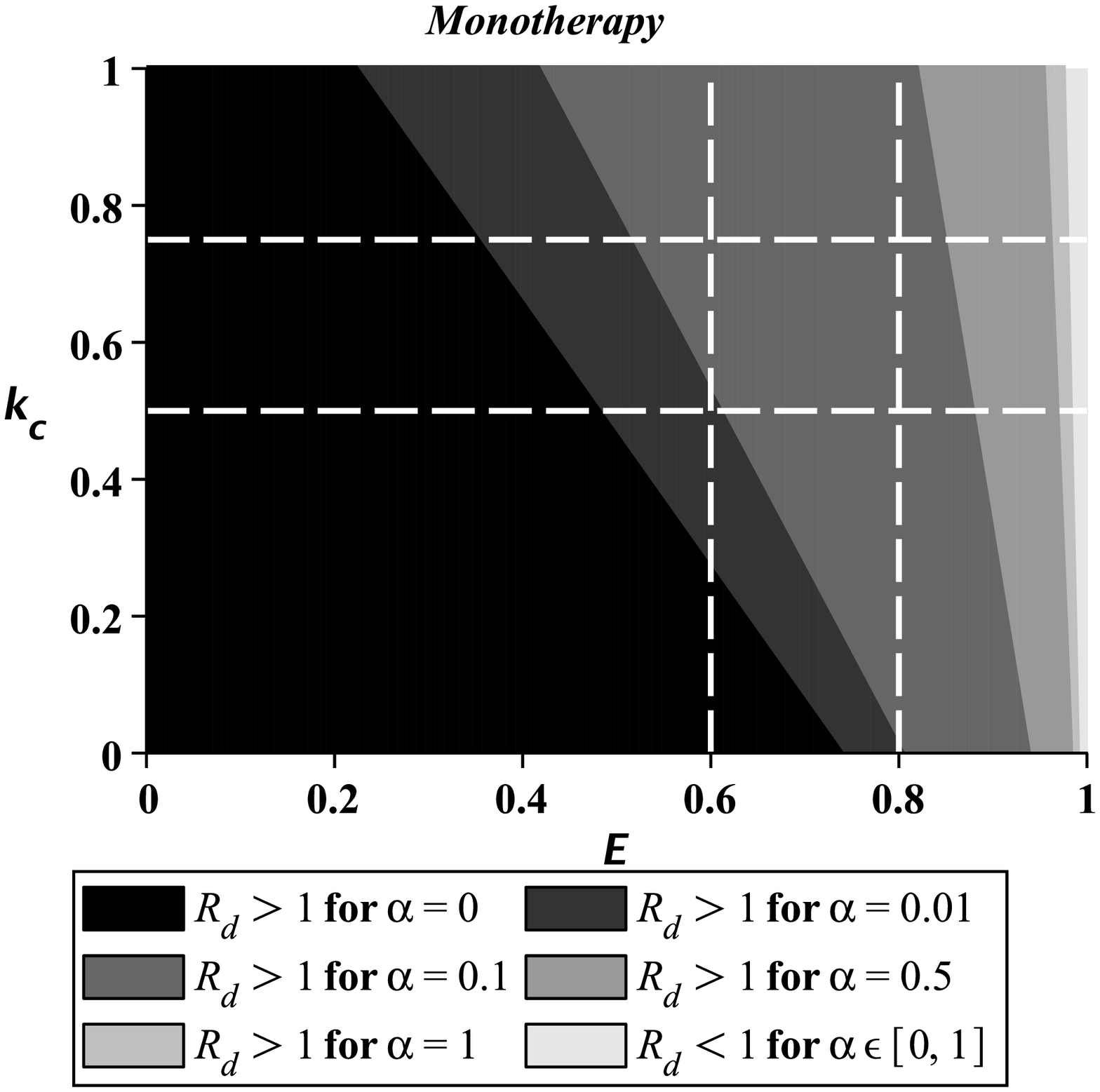}&
   \includegraphics[scale=.3]{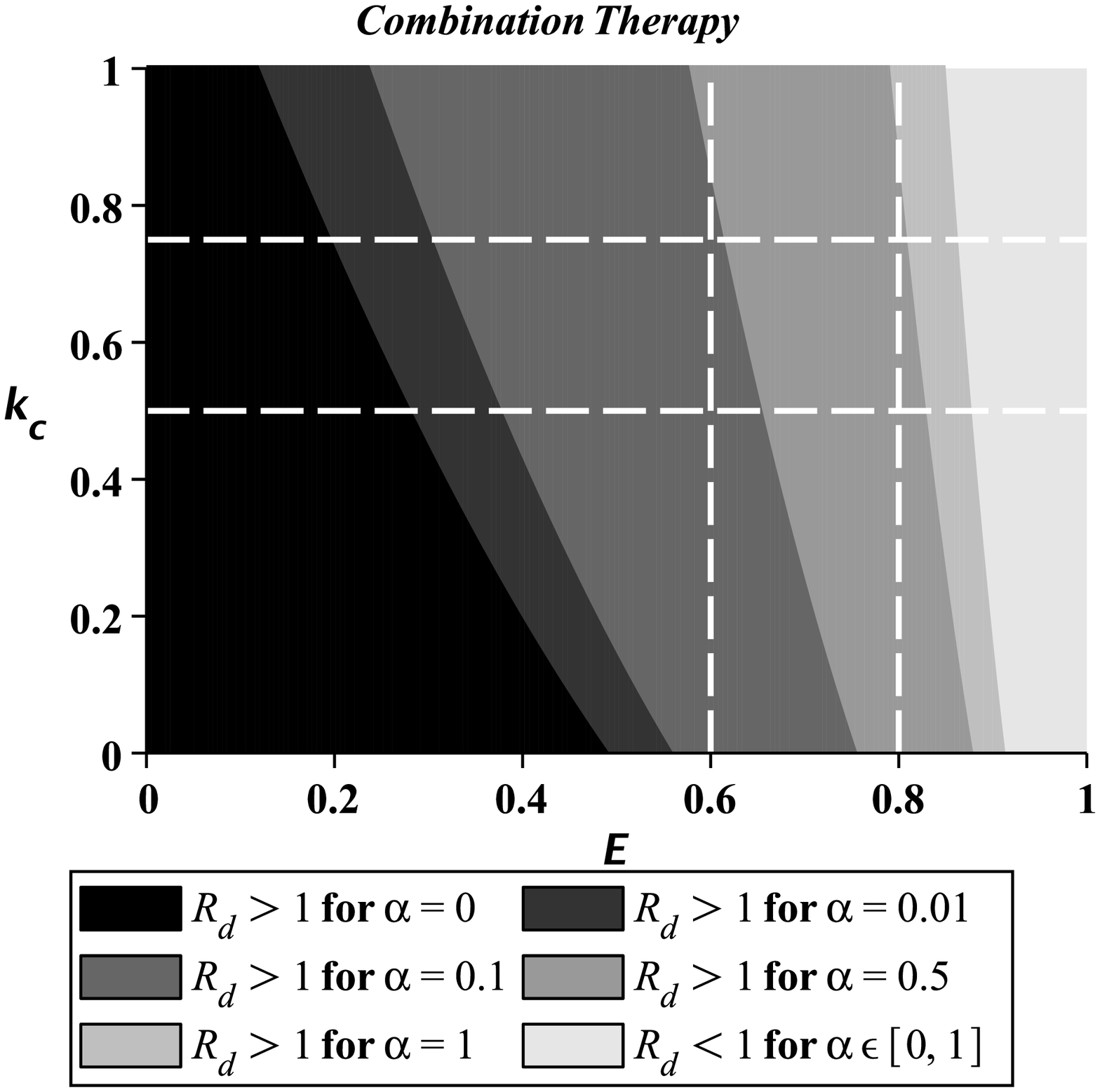}
\end{tabular}
\end{center}
\caption{\footnotesize A comparison between effectiveness of antiretroviral monotherapy  versus combination therapy using the quantity threshold $\mathcal{R}_d=\frac{s\,e_2\,e_4}{\sigma\,\gamma_1\,e_3}$ is studied in this experiment related to the multiple sets of variables $(E_{RT},\,E_{PI},\,\alpha)$. The horizontal and vertical axes are used for monotherapy or combination therapy, and the effect of CTL response respectively.  Based on simulations and theoretical results, an effective treatment strategy for the removal of integrated HIV cells is identified by the administration of combination therapy. From both experiments, we observe that the possibility of survival of viral load using monotherapy is more than $99\%$ for some HIV-infected patients while it can be broadly reduced to $80\%$ by intervention of both $E_{RT}$ and $E_{PI}$.}
 \label{chap3_pic1}
\end{figure}

\begin{table}[hbt]
\tikzstyle{int}=[rectangle, draw, fill=blue!10,
    text width=22em, text centered, rounded corners, minimum height=10.5em]
\tikzstyle{init} = [pin edge={to-,thick,black}]
{
\tikzstyle{cloud} = [draw=none, rectangle,text centered,rounded corners,text width=\textwidth,fill=gray!20, node distance=3.8cm,
    minimum height=5em]%;
\begin{tikzpicture}[node distance=2.cm,auto,>=latex']
    \node [cloud] (dd)
    {\caption{\footnotesize Effect of combination therapy: steady-state values \cite{chap7_A.B. Gumel}}\label{chap7_table2}
\begin{center}\footnotesize
\begin{tabularx}{\textwidth}{XXX}
\toprule %
$E_{PI}$ & $E_{RT}$& $\bar V_I$\\\hline
0.2 & 0.2 & 356.44 \\
%0.4 & 0.4 & 304.55 \\
0.6 & 0.6 & 282.69 \\
0.99 & 0.99 & 0 \\
[1ex]
\end{tabularx}
\end{center}};
\end{tikzpicture}
}
\end{table}

%----------------------------------------------------------------------------------
\subsection{Identification of regions for incompatible estimated densities} Our objective here is to identify regions where solutions generated from the numerical estimation for the correlated with dynamic discrete model is inconsistent with the behavior of solutions of the continuous HIV infection model. Based on the description of parameters in HIV model and their estimated values presented in Table \ref{chap7_table0}, incompatible iterative solutions is most likely to be generated in the regions when the key parameter $e_1$ is positive. As discussed previously, the estimated parameter $c$ can hold negative values conjunction with positive values for $e_1$.

%
%incompatible with the behavior of solutions of the continuous HIV infection model.
%
%can be identified for
%
%the possibility of to find the region where the iterative is .
%
% Description of parameters Inference in with

In order to get more insight into the effects of the key parameter $c$ on the perturbation of the solution, we
conducted another experiment using the function $E_{RT}=\frac{\alpha}{\alpha+r}$ for the effectiveness of reverse transcriptase inhibitors proportional to  the pre-existing activated CD4$^{+}$ T cells, $\alpha$.  In Figure \ref{chap3_pic2}, we observe that the effectiveness of reverse transcriptase inhibitors decline as the magnitude of the rate of proliferation of CD4$^{+}$ T cells is increased from 0 to 1. Hence this can significantly reduce the area where $e_1$ can be positive shrinking the regions of incompatible iterative solutions for numerical estimations.

\begin{figure}[b]
\begin{center}
   {\includegraphics[scale=.37]{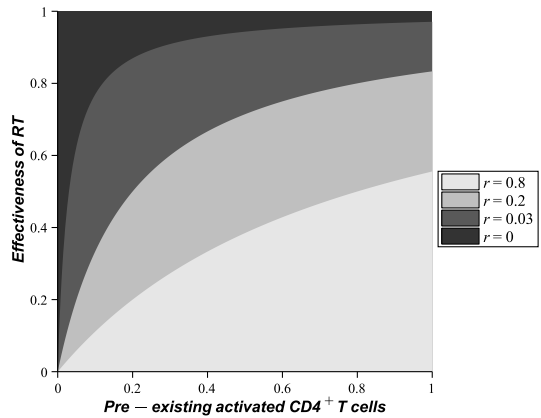}}
\end{center}
\caption{\footnotesize A comparison analysis for different values of the rate of proliferation of CD4$^{+}$ T cells. The parameter $e_1$ ($c$, respectively) can be positive (negative) in the area delimited by each curve, $E_{RT}-$axis and the constant value $E_{RT}=1$. Different curves are generated by the function $E_{RT}=\frac{\alpha}{\alpha+r}$ with respect to the proportion of pre-existing activated CD4$^{+}$ T cells, $\alpha$,  for $r=0,0.03,0.2,0.8$, respectively.}
 \label{chap3_pic2}
\end{figure}
%-------------------------------------------------------------------------------------
%-------------------------------------------------------------------------------------
\section{Conclusion}
This paper presents a method for analogy in continuous HIV infection model, system \eqref{chap7_e2.1}--\eqref{chap7_e2.3}, with estimation of correlated dynamic discrete model, system
(\ref{chap7_e4.4})$-$(\ref{chap7_e4.6}). We primarily construct the compartmental model of HIV infection for three different cell types, i.e. susceptible, infected, and infectious categories. The model simulates the interaction between CD4$^{+}$ T cells and HIV in vivo when combination antiretroviral therapy is used for the management of infection. Basically, reverse transcriptase inhibitors and protease inhibitors are exerted for the perturbation of HIV. Developing the robust numerical method \ref{e1.1}, we can successfully identify dynamics of behavioral changes in solutions of non-denationalized system (\ref{chap7_e4.4})$-$(\ref{chap7_e4.6}) by the quantity threshold $R_d$. A complete proof for the  convergence of the proposed estimation solutions
to the endemic equilibrium point is obtained for the Open problem given in  \cite[Section 4.3]{chap7_DNJ} with the assumptions imposed to the set of parameter values. The findings show that the burden of characterizing dynamics of the continuous HIV infection model \eqref{chap7_e2.1}--\eqref{chap7_e2.3} can be reduced to the computational discrete model at each iteration by efficient use of the estimation algorithm \ref{e1.1} and the information obtained
on previous stages.

Next a reliable comparison for the original deterministic HIV infection model with correlated dynamic discrete model is studies to analyze theoretically the long-term and oscillatory behavior of solutions. The results demonstrate that  the estimation accuracy of the proposed numerical method is an efficient algorithm for solving the dynamic continuous model in conjunction with $e_1 \le 0$, since similar results for both models are achieved. The main difference turns out to be in the qualitative behaviors of solutions of the models when $e_1>0$. Theoretical results and extensive experiments on artificial data with multiple sets of parameter values show that the discrete model conducted by the computational method conducted can behave quite different with the continuous deterministic model with HIV infection when $e_1>0$. Roughly speaking, the experiments demonstrate that  system (\ref{chap7_eq2}) can attain bounded or unbounded solutions \cite{chap7_M. Dehghan 2007} while the qualitative behavior of system (\ref{chap7_e4.4})$-$(\ref{chap7_e4.6}) with different initial values can lead to serious misspecification errors referred to as forbidden sets, see Appendix A.
In analysis of system (\ref{chap7_e4.4})$-$(\ref{chap7_e4.6}),  the paper establishes the complete convergence of the solutions to the nonnegative equilibrium points depending on the quantity threshold $R_d$; it then follows that no bistability occurs to the solutions and a forward bifurcation  is observed when the equality
$e=(b-1)(1-d)(1-f)$ hold.

The theoretical results follow by implementation of several experiments with respect to numerical algorithm to verify, among other things, the  effect of multiple parameter values for time-step and immune response. For instance, we show that the numerical method was a very robust and efficient technique in terms of stability to solve the continuous HIV infection system for large/small time steps since the equilibria and the stability conditions are independent of the time step. This was claimed without theoretical results by Gumel et al. \cite{chap7_A.B. Gumel}. Numerical simulations also indicate that transitory oscillations are observed for about the first 100 time units and no sustained oscillations are observed when the endemic equilibrium point appears to be stable for lower values of $R_d$ close to unity but not for significantly large values. However, the time periods of oscillations to the solutions of the model become shorter and rapid
in conjunction with occurrence of higher amplitudes of oscillations when the threshold level of basic reproduction number exceeds unity and is maintained above a certain threshold sufficiently large. Besides, we observe that the  qualitative oscillatory behavior of the correlated dynamic discrete model is not substantially perturbed by the implementation of various parameter values for $c$.

%This result seems to have a simple explanation. The expected long-term solution of the model (\ref{chap7_e4.4})$-$(\ref{chap7_e4.6}) associated with the equilibrium point \eqref{chap7_e4.8} is intuitively sensitive to the implementation of the parameter $c$ and the basic reproduction number. The larger the parameter value $c$, the higher amplitude is observed for the density of infected cells and virus particles, Figure \ref{chap3_pic3}. However the parameter value $c$ produces directly no sizeable effects or possible repercussions to the density of uninfected cells.

Further computational experiments are performed for the within-host infectious model in order to identify an effective treatment strategy for the control of HIV infection by application of single/multiple anti-HIV prophylactic medications within a patient. An intuitive comparison for both combination antiretroviral therapy and monotherapy is described for several case studies using multiple sets of variables $(E_{RT},\,E_{PI},\,\alpha)$. We subsequently study the possibility of successful persistence or inhibition of infectious viral load. In fact,  we conduct these experiments to investigate the qualitative effect of various proportion values of pre-exiting activated CD4$^{+}$ T cells  through perturbation of dynamics of the HIV model. To control and bring down the volume of viral load for monotherapy scenarios, the findings suggest that, in some cases, it is impossible to mitigate HIV risk in vivo, or it is required to monitor antiretroviral treatment that is $99\%$ effective or more. However the consequence of administrating two drugs is a success for mitigation and prevention of infectious virus load. Especially, if $\alpha$ is close to one, the virus particles can be significantly inhibited in the presence of only $90\%$ effective combination therapy.

\begin{appendix}
  \section{Dynamics of solutions with negative parameter $c$}\label{sec6}
In this section, $c$ is considered a negative number. As we will
see in the next section, this situation is equivalent to  $e_1>0$ in the
continuous model.

In this case we have again two equilibrium points (\ref{chap7_e4.7}) and
(\ref{chap7_e4.8}). But the equilibrium point (\ref{chap7_e4.8}) has negative
components, i.e. $\bar{y}$ and $\bar{z}$, while (\ref{chap7_e4.8})  is locally
asymptotically stable (Theorem \ref{chap7_t4.5}). However, if we can show that the solution eventually enters a negative area of $\mathbb{R}^3$, then viral load can be suppressed from the host.

%%%%%%%%%%%%%%%%%%%%%%%%%%%%%%%%%%%%%%%%%%%%%%%%%%%%%%%%%%%%%%%%%%%%%%%%%%%%%

\begin{Thm}
The following statements are true:
\begin{enumerate}[(i)]
  \item Assume that $x_0\ge-1$, $y_0=0$, and $z_0=0$, then $x_n>0$, $y_n=0$, and $z_n=0$.

  \item Assume that $x_0\ge-1$, $y_0<0$, and $z_0<\frac{-b}{c}$, then $x_n\ge0$, $y_n<0$, and $z_n<0$.

\end{enumerate}
\end{Thm}
\begin{proof}
The proof is an obvious result of  system  (\ref{chap7_e4.4})$-$(\ref{chap7_e4.6}).
\end{proof}

Now if we multiply the second and third equations of system  (\ref{chap7_e4.4})$-$(\ref{chap7_e4.6}) by $-1$, the change of variables $X_n=x_n$, $Y_n=(-y_n)$, $Z_n=(-z_n)$ gives
system (\ref{chap7_e4.4})$-$(\ref{chap7_e4.6}) with $c$ substituted by $(-c)$ which is positive in the first equation, and then we have
$\{(X_n,Y_n,Z_n)\}_{n=1}^{\infty}\subset\mathbb R_+^{3}$. As a result, the dynamic of the solution when $c$ is negative
with $x_0\ge-1$, $y_0<0$, and $z_0<\frac{-b}{c}$ is exactly the same
as the positive case.

For other situations, it is not possible to verify the dynamics of
solutions, since it is difficult to find the forbidden
set {\bf F}, the set of initial conditions $(x_0,y_0,z_0)\in
\mathbb{R}^3$ through which the denominator $(b+cz_{n})$ in equation
(\ref{chap7_e4.4}) will become zero for some value of $n\ge0$. Note that the situation $c<0$, $x_0\ge-1$, $y_0<0$, and
$z_0<\frac{-b}{c}$ in (\ref{chap7_e4.4})$-$(\ref{chap7_e4.6}) corresponds to
$e_1>0$, $x_0\geq 0$, $y_0<0$, and $z<0$  in the continuous model  (\ref{chap7_eq2}).

However we can leave a conjecture here that with an initial solution in this
form $(x_0,y_0,z_0)\in \mathbb{R}^3 \backslash {\bf F}$, the solution will become eventually as follows:
\[x_n\ge0, \quad y_n\le0, \quad \mbox{and} \quad z_n\le0.\]

\subsection{Sensitivity analysis for negative parameter experiments}
\begin{Exm}\label{chap7_ex6.1}
In order to test the stability and convergence properties of the scheme constructed in
Sections 5.4, 5.5,  and 5.6, we use the method to simulate the model (\ref{chap7_e4.4})$-$(\ref{chap7_e4.6}) with the  initial values and parameter
$(x_0,y_0,z_0)=(1,0,10^{-10})$,  $b=2$, $c=-1$, $d=0.5$, $e=1$, and $f=0.5$.
\end{Exm}
We have clearly that $e>(b-1)(1-d)(1-f)$, therefore, the equilibrium
point, $(0.25,-1.5,-3)$, is locally asymptotically stable. Table
\ref{chap7_table3} illustrates the global asymptotical stability of the
equilibrium point $(0.25,-1.5,-3)$.

%---------------------------------------------------------------------------------------------

\begin{table}[hbt]
\tikzstyle{int}=[rectangle, draw, fill=blue!10,
    text width=22em, text centered, rounded corners, minimum height=10.5em]
\tikzstyle{init} = [pin edge={to-,thick,black}]
{
\tikzstyle{cloud} = [draw=none, rectangle,text centered,rounded corners,text width=\textwidth,fill=gray!20, node distance=3.8cm,
    minimum height=5em]%;
\begin{tikzpicture}[node distance=2.cm,auto,>=latex']
    \node [cloud] (dd)
    {\caption{\footnotesize Negative parameter together with positive initial values leads to negative solutions, that is, a small portion of viral load can be suppressed immediately.}\label{chap7_table3}
\begin{center}
{\footnotesize\begin{tabularx}{\textwidth}{llXXl}
\toprule %
  $n$    && Density of $x_{n}$                   & Density of $y_{n}$                   & Density of $z_{n}$   \\ \hline
  $0$    && $1$                       & $0$                       & $10^{-10}$        \\
  $10$   && $1.000000015$             & $0.2955390683\times10^{-7}$    &  $0.4037138739\times10^{-7}$\\
%  $20$   && $1.000007565$             & $0.1512847059\times10^{-4}$    & $0.2066584326\times10^{-4}$\\
%  $30$   && $1.003893184$             & $0.7786367848\times10^{-2}$    & $0.1062795489\times10^{-1}$\\
%  $40$   && $-0.3290565393\times10^{-1}$   & $-2.065811307$            & $-2.549563692$\\
  $50$   && $0.2502325578$            & $-1.499534885$            & $-2.997932087$\\
%  $60$   && $0.2500028949$            & $-1.499994211$            & $-2.999974241$\\
%  $70$   && $0.2500000361$            & $-1.499999928$            & $-2.999999679$\\
%  $80$   && $0.2500000005$            & $-1.500000000$            & $-2.999999997$\\
%  $90$   && $0.2500000000$            & $-1.500000000$            & $-3.000000000$\\
  $100$  && $0.2500000000$            & $-1.500000000$            & $-3.000000000$\\
   [1ex]
\end{tabularx}}
\end{center}};
\end{tikzpicture}
}
\end{table}

%---------------------------------------------------------------------------------------------

%---------------------------------------------------------------------------------------------

\begin{Exm}\label{chap7_ex6.2}
In this example, we assume $b=5$, $c=-1$, $d=0.2$, $e=1$, and
$f=0.5$. We consider two solutions of (\ref{chap7_e4.4})$-$(\ref{chap7_e4.6})
with respect to two initial values $(0.4,0.75,1.5+10^{-5})$ and
 $(0.4,0.75,1.5-10^{-5})$.
\end{Exm}
We observe that $e<(b-1)(1-d)(1-f)$, therefore, the equilibrium
point, $(0.25,0,0)$, is locally asymptotically stable. Tables \ref{chap7_table4} and
\ref{chap7_table5} illustrate the global asymptotical stability of the
equilibrium point $(0.25,0,0)$. With the initial value
$(0.4,0.75,1.5+10^{-5})$, $y_{n}$ and $z_{n}$ are negative after
finite iterations. While for the initial value $(0.4,0.75,1.5-10^{-5})$, each component of $(x_{n},y_{n},z_{n})$ is
positive.
\begin{table}[hbt]
\tikzstyle{int}=[rectangle, draw, fill=blue!10,
    text width=22em, text centered, rounded corners, minimum height=10.5em]
\tikzstyle{init} = [pin edge={to-,thick,black}]
{
\tikzstyle{cloud} = [draw=none, rectangle,text centered,rounded corners,text width=\textwidth,fill=gray!20, node distance=3.8cm,
    minimum height=5em]%;
\begin{tikzpicture}[node distance=2.cm,auto,>=latex']
    \node [cloud] (dd)
    {\caption{\footnotesize Positive initial values and  negative parameter converges eventually to zero but not inside the first octant.}\label{chap7_table4}
\begin{center}
{\footnotesize\begin{tabularx}{\textwidth}{llXXl}
\toprule %
  $n$     & & Density of $x_{n}$          & Density of $y_{n}$                   & Density of $z_{n}$  \\ \hline
  $0$     & & $0.4$            & $0.75$                    & $1.5+10^{-5}$             \\
  $20$    & & $0.4000775562$   & $0.7503877812$            & $1.500649000$\\
%  $40$    & & $0.4061904916$   & $0.7809524575$            & $1.551541362$\\
%  $60$    & & $0.2407665726$   & $-0.4616713744\times10^{-1}$   & $-0.1183643855$\\
%  $80$    & & $0.2497903198$   & $-0.1048400521\times10^{-2}$   & $-0.2640833921\times10^{-2}$\\
  $100$   & & $0.2499950251$   & $-0.2487438382\times10^{-4}$   & $-0.6262803901\times10^{-4}$\\
%  $120$   & & $0.2499998820$   & $-0.5907468859\times10^{-6}$   & $-0.1487350223\times10^{-5}$\\
%  $140$   & & $0.2499999972$   & $-0.1403009524\times10^{-7}$   & $-0.3532419798\times10^{-7}$\\
%  $160$   & & $0.2500000000$   & $-0.3332115421\times10^{-9}$   & $-0.8389415901\times10^{-9}$\\
%  $180$   & & $0.2500000000$   & $-0.7913697756\times10^{-11}$  & $-0.1992467049\times10^{-10}$\\
  $200$   & & $0.2500000000$   & $-0.1879485081\times10^{-12}$  & $-0.4732063579\times10^{-12}$\\
   [1ex]
\end{tabularx}}
\end{center}};
\end{tikzpicture}
}
\end{table}
\begin{table}[hbt]
\tikzstyle{int}=[rectangle, draw, fill=blue!10,
    text width=22em, text centered, rounded corners, minimum height=10.5em]
\tikzstyle{init} = [pin edge={to-,thick,black}]
{
\tikzstyle{cloud} = [draw=none, rectangle,text centered,rounded corners,text width=\textwidth,fill=gray!20, node distance=3.8cm,
    minimum height=5em]%;
\begin{tikzpicture}[node distance=2.cm,auto,>=latex']
    \node [cloud] (dd)
    {\caption{\footnotesize Negative parameter and  positive initial values leads to positive solutions, that is, the viral load can still persist. One may compare this result with Table \ref{chap7_table4}.}\label{chap7_table5}
\begin{center}
{\footnotesize\begin{tabularx}{\textwidth}{llXXl}
\toprule %
   $n$     & & Density of $x_{n}$          & Density of $y_{n}$                  & Density of  $z_{n}$  \\ \hline
    $0$     & & $0.4$             &  $0.75$                    & $1.5-10^{-5}$ \\
    $20$    & & $0.3999225278$    &  $0.7496126391$            & $1.499351619$\\
%    $40$    & & $0.3943168060$    &  $0.7215840306$            & $1.452209971$\\
%    $60$    & & $0.2912673373$    &  $0.2063366867$            & $0.4790217445$\\
%    $80$    & & $0.2512455549$    &  $0.6227774544\times10^{-2}$    & $0.1563684572\times10^{-1}$\\
    $100$   & & $0.2500297591$    &  $0.1487952268\times10^{-3}$    & $0.3746036600\times10^{-3}$\\
%    $120$   & & $0.2500007068$    &  $0.3534354679\times10^{-5}$    & $0.8898588764\times10^{-5}$\\
%    $140$   & & $0.2500000168$    &  $0.8394039679\times10^{-7}$    & $0.2113404802\times10^{-6}$\\
%    $160$   & & $0.2500000005$    &  $0.1993565348\times10^{-8}$    & $0.5019288560\times10^{-8}$\\
%    $180$   & & $0.2500000000$    &  $0.4734672028\times10^{-10}$   & $0.1192069535\times10^{-9}$\\
    $200$   & & $0.2500000000$    &  $0.1124473755\times10^{-11}$   & $0.2831137823\times10^{-11}$\\
   [1ex]
\end{tabularx}}
\end{center}};
\end{tikzpicture}
}
\end{table}

\end{appendix}
\end{document}